%% file: t_main.tex
\documentclass[11pt]{article}

\usepackage[round]{natbib}
\usepackage{amsmath,amsthm,enumitem,nicefrac,bbm,verbatim,amssymb,amsfonts,amscd, graphicx,algpseudocode,hyperref,float,mathtools,bm,xcolor,appendix,subcaption,diagbox,caption,multirow,makecell}
\usepackage{tikz}
\usepackage{tikz-3dplot}
\usetikzlibrary{positioning, decorations.pathreplacing, calc, intersections, pgfplots.groupplots}
\usetikzlibrary{positioning, decorations.markings}
\usepackage{pgfplots}
\usepackage{tikz-network}
\usetikzlibrary{shapes,decorations,arrows,calc,arrows.meta,fit,positioning}

\pgfplotsset{compat=1.5}
\pgfplotsset{scaled x ticks=false}
\tikzset{
	-Latex,auto,node distance =1 cm and 1 cm,semithick,
	state/.style ={ellipse, draw, minimum width = 0.7 cm},
	point/.style = {circle, draw, inner sep=0.04cm,fill,node contents={}},
	bidirected/.style={Latex-Latex},
	el/.style = {inner sep=2pt, align=left, sloped}
}
\usepackage[utf8]{inputenc}
\usepackage{authblk}
\usepackage[T1]{fontenc}
\usepackage[ruled, lined, linesnumbered, commentsnumbered, longend]{algorithm2e}

\input{glodef}
\input{locdef}

\usepackage{utopia}
\usepackage{xifthen}
\usepackage{soul}
\usepackage{thm-restate}
\usepackage{hyperref}
\newcommand*{\SINGLE}{}
\title{Testing Partially-Identifiable Causal Queries Using Ternary Tests}
\author[1]{Sourbh Bhadane} 
\author[1]{Joris M. Mooij}
\author[2]{Philip Boeken}
\author[3]{Onno Zoeter}
\affil[1]{Korteweg-de Vries Institute for Mathematics, University of Amsterdam}
\affil[2]{Department of Mathematics, VU Amsterdam}
\affil[3]{Booking.com, Amsterdam, The Netherlands}
\begin{document}

\maketitle
\begin{abstract}
  We consider hypothesis testing of binary causal queries using observational data. Since the mapping of causal models to the observational distribution that they induce is not one-to-one, in general, causal queries are often only partially identifiable. When binary statistical tests are used for testing partially-identifiable causal queries, their results do not translate in a straightforward manner to the causal hypothesis testing problem. We propose using ternary (three-outcome) statistical tests to test partially-identifiable causal queries. We establish testability requirements that ternary tests must satisfy in terms of uniform consistency and present equivalent topological conditions on the hypotheses. To leverage the existing toolbox of binary tests, we prove that obtaining ternary tests by combining binary tests is complete. Finally, we demonstrate how topological conditions serve as a guide to construct ternary tests for two concrete causal hypothesis testing problems, namely testing the instrumental variable (IV) inequalities and comparing treatment efficacy. 
\end{abstract}

\input{intro.tex}
\input{related.tex}
\input{setup}

\input{preliminaries}

\input{math_ternary.tex}

\input{iv_three_arm}
\input{manski}
\input{colliders}
\input{numerical}

\input{discussion.tex}
\bibliographystyle{apalike}
\bibliography{biblio.bib}
\newpage
\appendix
\input{appendix}
\end{document}

%% file: glodef.tex
\newtheorem{Theorem*}{Theorem}

\newtheorem{Claim*}[Theorem]{Claim}

\newtheorem{CounterExample*}{$\overline{\hbox{\bf Example}}$}

\newtheorem{Example*}[Theorem]{Example}

\newtheorem{Intuition*}[Theorem]{Intuition}
\newtheorem{Joke*}[Theorem]{Joke}

\newtheorem{Lemma*}[Theorem]{Lemma}
\newtheorem{Open problem}[Theorem]{Open problem}

\newtheorem{Question*}[Theorem]{Question}

\usepackage{fullpage}

\newtheorem{theorem}{Theorem}
\newtheorem{proposition}[theorem]{Proposition}

\newtheorem{lemma}[theorem]{Lemma}

\newtheorem{definition}[theorem]{Definition}

\def \bSubexa    {\begin{subexa}}

\usepackage[colorinlistoftodos,textsize=scriptsize]{todonotes}

\newcommand{\ignore}[1]{}

\newcommand{\NN}{\mathbb{N}}

\newcommand{\RR}{\mathbb{R}}

\newcommand{\MM}{\mathbb{M}}

\def \cB     {{\cal B}}

\def \cP     {{\cal P}}

\def \cX     {{\cal X}}
\def \cY     {{\cal Y}}
\def \cZ     {{\cal Z}}

\def \Paren#1{{\left({#1}\right)}}

\def \Brack#1{{\left[{#1}\right]}}

\def\ignore#1{}

\newcommand{\bi}{\begin{itemize}}
\newcommand{\ei}{\end{itemize}}

\def\orpro{\mathop{\mathchoice
   {\vee\kern-.49em\raise.7ex\hbox{$\cdot$}\kern.4em}
   {\vee\kern-.45em\raise.63ex\hbox{$\cdot$}\kern.2em}
   {\vee\kern-.4em\raise.3ex\hbox{$\cdot$}\kern.1em}
   {\vee\kern-.35em\raise2.2ex\hbox{$\cdot$}\kern.1em}}\limits}

\def\andpro{\mathop{\mathchoice
 {\wedge\kern-.46em\lower.69ex\hbox{$\cdot$}\kern.3em}
 {\wedge\kern-.46em\lower.58ex\hbox{$\cdot$}\kern.25em}
 {\wedge\kern-.38em\lower.5ex\hbox{$\cdot$}\kern.1em}
 {\wedge\kern-.3em\lower.5ex\hbox{$\cdot$}\kern.1em}}\limits}

\def\simge{\mathrel{%
   \rlap{\raise 0.511ex \hbox{$>$}}{\lower 0.511ex \hbox{$\sim$}}}}

\def\simle{\mathrel{
   \rlap{\raise 0.511ex \hbox{$<$}}{\lower 0.511ex \hbox{$\sim$}}}}

\providecommand{\email}[1]{\href{mailto:#1}{\nolinkurl{#1}\xspace}}



%% file: locdef.tex
\newcommand{\model}{M}

\newcommand{\doop}[1]{\text{do}(#1)}

\newcommand{\enop}{V}
\newcommand{\exrv}{W}
\newcommand{\spc}{\mathcal{X}}

\newcommand{\cg}[1]{G\Paren{#1}}

\newcommand{\modelivrelax}{\MM_{\text{IV}}}

\newcommand{\distq}{\cP_{H^0_{Q}}}

\newcommand{\distqalt}{\cP_{H^1_{Q}}}
\newcommand{\nullcausalq}{H^0_{Q}}

\newcommand{\altcausalq}{H^1_{Q}}

\newcommand{\nulliv}{H^0_{\text{IV}}}
\newcommand{\altiv}{H^1_{\text{IV}}}
\newcommand{\distivnull}{\cP_{H^0_{\text{IV}}}}
\newcommand{\distivalt}{\cP_{H^1_{\text{IV}}}}

\newcommand{\modelsivedge}{\MM_{\text{IV,edge}}}

\newcommand{\nll}{H^0}
\newcommand{\alt}{H^1}
\newcommand{\idk}{2}
\newcommand{\test}{\phi}
\newcommand{\fone}{\alpha}
\newcommand{\fzero}{\beta}
\newcommand{\param}{P}
\newcommand{\TT}{\text{TT}}
\newcommand{\prt}{R}
\newcommand{\sn}{\text{SN}}
\newcommand{\sa}{\text{SA}}

\newcommand{\world}{W}


%% file: intro.tex
\section{Introduction}\label{sec:intro}


A causal query is deemed identifiable from observational data whenever causal models that induce the same observational distribution also agree on the query of interest. For hypothesis testing of binary causal queries, such as testing the existence of an edge in a causal graph or testing whether a causal effect is $0$ or not, identification allows us to translate the results of a statistical test based on observational data to implications about causal hypotheses. For identifiable binary causal queries, the set of observational distributions induced by causal models where the binary causal query evaluates to $0$, is disjoint from those induced by causal models where the binary causal query evaluates to $1$. However, often causal queries of interest are ``not identifiable'' from observational data because of reasons such as presence of unobserved confounders, positivity violations, etc. In such cases, the aforementioned sets of induced observational distributions overlap. If the overlap is partial, then the causal query is \textit{partially-identifiable}. If the sets completely overlap, the causal query is deemed \textit{non-identifiable} since every distribution could be induced by causal models corresponding to both values of the causal query. 

When binary statistical tests are used to test partially-identifiable causal queries, their results do not translate in a straightforward manner to the causal hypothesis testing problem. For example, consider a binary test based on observational data such that the statistical null hypothesis is the set of observational distributions induced by causal models where the causal query evaluates to $0$. In this case, rejecting the null implies that the causal query, evaluated on the underlying causal model, evaluates to $1$. However, not rejecting the null presents an identifiability issue since the null includes the non-empty intersection of the induced sets of observational distributions. Therefore, not rejecting the null only allows concluding ``unidentifiable''. See \citet{BhadaneMBZ25} for an example of a statistical test for the absence of a direct effect that leads to a similar conclusion.

In this work, we propose testing partially-identifiable causal queries using statistical tests that output one of three outcomes; ``reject the null'', ``don't reject the null'', and ``unidentifiable'' where the former two correspond to the set differences of the induced observational distributions and the latter corresponds to the intersection. Clearly, a ternary test provides more fine-grained information than a binary test. In addition, translating the result of a ternary test to the causal hypothesis testing problem is straightforward; if the ternary test outputs either ``reject the null'' or ``don't reject the null'', the conclusion for the causal hypotheses is exact, otherwise it is unidentifiable from observational data. 

We extend the binary hypothesis testing framework to testing overlapping hypotheses with three outcomes which can be viewed as testing three disjoint hypotheses with three outcomes. We consider an error to be when the output of a ternary test is different from the hypothesis that contained the distribution that samples were drawn from. We first define error-guarantee desiderata for ternary tests. In binary hypothesis testing, these are made in terms of asymptotic consistency and error-control \citep{DemboPeres94, Ermakov17, GeninKelly17, Genin18, BoekenSGM25} and then related to topological conditions on the hypotheses in a suitable topology of probability measures. We follow suit by proposing analogous equivalent topological conditions for existence of `good' ternary tests. Given the vast toolbox of binary tests developed so far, we study ternary tests that are constructed by combining binary tests, which we call a ``two-stage ternary test''. We show via topological conditions on the hypotheses that two-stage ternary testing is complete in the sense that it attains the same error control as that of ternary tests that control ``column-wise'' errors. We demonstrate how topological conditions on hypotheses can guide a practitioner in constructing ternary tests by considering two concrete applications to testing causal queries, namely testing the instrumental variable (IV) inequalities \citep{Pearl95} expressed as conditions on the interventional Markov kernel using observational data, and comparing the treatment efficacy of a treatment from observational data against a threshold that might, for example, correspond to the estimated efficacy of a competing treatment obtained from a large randomized control trial. 

%% file: related.tex
\subsection{Related Work}
Ternary hypothesis testing has often been proposed in the context of mitigating shortcomings of the Neyman-Pearson null hypothesis significance testing framework. One of the earliest mentions can be found in the works of Neyman \citep{Neyman76}, Tukey \citep{Tukey91, JonesTukey00} and \citet{Hays63}. In particular, \citet{EstevesISS16} show that statistical tests that are logically coherent and optimal are ternary and not binary. \citet{Berg04} considers ternary tests for simple binary hypothesis testing and provide an optimal likelihood ratio test for the same. \citet{IzbickiCCLSS23} considers a region-based test that constructs confidence regions based on data and define a ternary test based on whether the confidence region is completely within the null, completely within the alternative or has non-empty intersection with both. \citet{BoekenSGM25} consider ternary tests to obtain Type-I and Type-II error control with finite-precision measurements. In sequential hypothesis testing, since a sample size is not determined apriori, at every stage, a sequential test decides between accepting the null, rejecting the null, and collecting another sample \citep{Wald45}. See also, \cite{Berger85} for a Bayesian analogue of sequential hypothesis testing and \cite{DassBerger03} for examples of ternary tests in Bayesian testing of composite hypotheses. \citet{GarivierK21} consider testing $M$ overlapping hypotheses in the sequential testing framework using $M$-ary tests by running $M$ parallel binary tests with a restrictive notion of error. While some of the aforementioned works could be viewed as considering testing of overlapping hypotheses whereas others could be viewed as considering ternary tests for disjoint hypotheses, they only consider specialized cases and do not consider the question of existence of ternary tests that satisfy desiderata related to error-guarantees. In causal hypothesis testing, ternary tests were considered in passing by \citet{RobinsSSW03, ZhangSpirtes02}. In addition, ternary tests have been proposed in the specific context of detecting colliders in causal discovery algorithms \citep{RamseySZ06, ZhangSpirtes08} but our work provides a general analysis of ternary tests proving their existence under topological conditions and a way to construct ternary tests by combining binary tests. 





%% file: setup.tex
\section{Partial Identification: Need for Ternary Tests}\label{sec:setup}
We formalize the need for ternary tests, as mentioned in Section~\ref{sec:intro}, under the semantic framework of structural causal models (SCMs) and provide definitions related to causal models in Appendix~\ref{app:prelim_causal}. Consider a family of causal models $\MM$ and a binary causal query, $Q : \MM \mapsto \left \lbrace 0,1 \right \rbrace$. Define the causal null hypothesis 
\begin{equation}\label{eq:causalnull}
\nullcausalq \triangleq \left \lbrace \model \in \MM : Q(\model) = 0 \right \rbrace, 
\end{equation}
where the causal alternative hypothesis is defined by $\altcausalq \triangleq \MM \setminus \nullcausalq$. Given that we often only have access to observational data, i.e., data sampled from $P_{\model}\Paren{X_{\enop}}$, the set of observational distributions induced by the causal null hypothesis is given by $$\distq \triangleq \left \lbrace P_{\model} : \model \in \nullcausalq \right \rbrace,$$ and the set of observational distributions induced by the causal alternative hypothesis is given by $\distqalt \triangleq \left \{ P_{\model} : \model \in \altcausalq\right \}.$

If $Q$ is partially-identifiable, then $\distq \cap \distqalt \neq \emptyset$, i.e., while the causal null $\nullcausalq$ and the corresponding alternative $\altcausalq$ are disjoint, the resulting sets of distributions that they induce, overlap, in general. When a binary-outcome statistical test with null hypothesis $\distq$ and disjoint alternative hypothesis $\world \backslash \distq$ is used where $\world \triangleq \distq \cup \distqalt$,\footnote{Note that since $\distqalt \cap \distq$ is non-empty, $\distqalt$ cannot be a disjoint alternative hypothesis.} the results of this statistical test have the following implications for the corresponding causal hypotheses: rejecting the null, implies rejecting the causal null, however, not rejecting the null does not imply that the underlying causal model is in the causal null since $\distq \cap \distqalt \subseteq \distq $ implying that given no apriori information, a distribution in $\distq$ could have been induced by a causal model in $\altcausalq$. A similar observation was also made by \citet{ZhangSpirtes02} and more recently in \citet{BhadaneMBZ25} where the relevant causal query was the absence of a direct effect. For partially-identifiable causal queries, since $\distq$ and $\distqalt$ could overlap in general, we propose testing such queries using ternary tests. 

A natural question that arises is how to construct a `good' ternary test for testing a given partially-identifiable binary causal query. At this stage we can decouple the context of partial identification and consider a) what are desiderata for a `good' ternary test, and b) whether such `good' ternary tests exists in the first place. We base our answer to the latter question on topological conditions that the statistical hypotheses must satisfy. While these topological conditions might seem restricted to the question of existence of ternary tests, we subsequently demonstrate how a practitioner could use them to guide their construction of a ternary test. Finally, we demonstrate this guide for constructing  ternary tests for the partial identification examples in Section~\ref{sec:applications}.


%% file: preliminaries.tex
\section{Preliminaries} \label{sec:prelim}
Denote $\cP\Paren{\cX}$ as the set of all probability measures defined on the Borel $\sigma$-algebra $\cB\Paren{\cX}$ where $\cX$ is a separable metric space. Hypotheses $\nll,\alt \subseteq \cP\Paren{\cX}$ are considered with respect to the weak topology on $\cP\Paren{\cX}$ defined by the notion of weak convergence of probability measures \citep{Billingsley94}. For simplicity, we assume that $\world \triangleq \nll \cup \alt$ is compact in the topology of setwise convergence on $\cP\Paren{\cX}$ which is defined by pointwise convergence of probability measures on all Borel measurable sets, thus making it finer than the weak topology and implying that $\world$ is compact in the weak topology on $\cP\Paren{\cX}$. 
 We present a few known results below for existence of binary tests for i.i.d. data with certain desiderata of error-guarantees. 
\begin{definition}[Weak and Strong Consistency: Binary Tests]
Given disjoint hypotheses $\nll, \alt \subseteq \mathcal{P}\Paren{\mathcal{X}}$, a binary test $\phi = \left \lbrace \test_n \right \rbrace_{n \in \NN}$, where $\phi_n : \mathcal{X}^n \mapsto \left \lbrace 0,1 \right \rbrace$ is measurable, is
\begin{enumerate}
    \item weakly consistent if for all $i \in \left \lbrace 0,1 \right \rbrace$,  for all $P \in H_i$, $ \lim\limits_{n \rightarrow \infty}  P^n\Paren{\test_n\Paren{X^n}=i} = 1$ ;
    \item strongly consistent if for all $i \in \left \lbrace 0,1 \right \rbrace$, for all $P \in H_i$, $ P^{\infty}\Paren{\test_n\Paren{X^n} \neq i \text{ for finitely many }n} = 1$. 
\end{enumerate}
\end{definition}

\citet{DemboPeres94} provide sufficient conditions on $\nll, \alt$ for the existence of a strongly consistent test which are necessary if in addition the existence of $p>1$-integrable densities for every measure in $\world$ is assumed. Under a weaker assumption of $\world$ being $\sigma$-compact, these conditions are proven to be necessary for the existence of a weakly consistent test \citep[Theorem 4.4]{Ermakov17} and also for the existence of a strongly consistent test \citep[Theorem 3.3]{Ermakov17}. We present relevant results from \citet{Ermakov17} under the aforementioned assumptions on $\cX, \world$ for disjoint $\nll, \alt$. See Appendix~\ref{app:sec:ermakov} for a proof. 
\begin{restatable}{proposition}{consistent}\label{prop:consistencyfsigma}
 The following are equivalent:
 \begin{enumerate}
     \item[(i)]there exists a binary test that is weakly consistent.
     \item [(ii)] $\nll$ and $\alt$ are disjoint and $F_{\sigma}$ (i.e., countable union of closed sets) in the weak topology on $\cP\Paren{\cX}$.
 \end{enumerate}
 \end{restatable}
Weak and strong consistency are notions of point-wise consistency. A stronger error-guarantee desideratum is that the above hold uniformly for all $P$ in relevant hypotheses. We will only consider the uniform analogue of weak consistency. 
\begin{definition}[Uniform Consistency: Binary Tests]
Given disjoint hypotheses $\nll, \alt \subseteq \mathcal{P}\Paren{\mathcal{X}}$, a binary test $\phi = \left \lbrace \test_n \right \rbrace_{n \in \NN}$, where $\phi_n : \mathcal{X}^n \mapsto \left \lbrace 0,1 \right \rbrace$ is measurable, is uniformly consistent if for all $i \in \left \lbrace 0,1 \right \rbrace$, $ \lim\limits_{n \rightarrow \infty} \sup\limits_{P \in H_i} P^n\Paren{\test_n\Paren{X^n}\neq i} = 0$.
\end{definition}
Necessary and sufficient conditions for existence of uniformly consistent binary tests were given by \citet{Berger51} that are technical and difficult to apply. Le Cam and Schwartz \citep{LeCamSchwartz60} provide the same in terms of the topology of setwise convergence on all $n$-fold products of probability measures, which are ``rather inaccessible'' \citep{LeCam86}.  Much of subsequent work has focused on introducing assumptions on $\world$ so that these conditions are formulated in the weak topology. See \citet{Kleijn22} for a comprehensive overview. The following is a rephrasing of \citet[Theorem 4.1]{Ermakov17} under aforementioned assumptions on $\cX, \world$ and disjoint $\nll,\alt$. See Appendix~\ref{app:sec:ermakov} for a proof.  
\begin{restatable}{proposition}{unifclopen}\label{prop:unifclopen}
 The following are equivalent:
 \begin{enumerate}
     \item[(i)]there exists a binary test that is uniformly consistent.
     \item [(ii)] $\nll$ and $\alt$ are disjoint and closed in the weak topology on $\cP\Paren{\cX}$.
 \end{enumerate}
 \end{restatable}

In similar vein, \citet{Ermakov17} also consider existence of tests such that uniform consistency holds only for the null and weak consistency holds for the alternative (and the null since uniform consistency holds). 

\begin{restatable}{proposition}{unifclosed}\label{prop:unifclosed}
The following are equivalent:
 \begin{enumerate}
     \item[(i)]there exists a weakly consistent binary test that is uniformly consistent with respect to $\nll$.
     \item [(ii)] $\nll$ is closed and $\alt$ is $F_{\sigma}$ in the weak topology on $\cP\Paren{\cX}$.
 \end{enumerate}
 \end{restatable}

An even stronger error-guarantee desiderata on a test is that the worst-case error for $\nll$ (or $\alt$) is bounded by a predetermined constant $\alpha$ (or $\beta$) for every sample-size $n$, often referred to as Type-I (or Type-II) error guarantee \footnote{Although \citet{BoekenSGM25} prove that in terms of existence of tests the above notion of finite-sample error control is equivalent to uniform consistency.}. Nonexistence of binary tests that control both the Type-I and Type-II error for arbitrary $\alpha$ or $\beta$ follows from LeCam's two point lower bound \citep{LeCam73}. \citet{GeninKelly17}  consider tests with critical regions that are continuity sets (stronger than binary tests that are measurable functions) under the assumption of existence of densities for every measure in $\world$ (weaker than compactness) and provides the same topological conditions as that of Proposition~\ref{prop:unifclosed} for existence of consistent (weakly and strongly) binary tests that control Type-I error for arbitrary $\alpha$. They also provide same topological conditions as that of Proposition~\ref{prop:unifclopen} for existence of consistent (with respect to $\nll$ and $\alt$) ternary tests that control both the Type-I and Type-II error for arbitrary $\alpha$ and $\beta$. \citet{BoekenSGM25} consider ternary tests with open critical regions (corresponding to decisions $0$ and $1$) and provide same topological conditions as that of Proposition~\ref{prop:unifclosed} for existence of consistent binary tests that control Type-I error for arbitrary $\alpha$ and same topological conditions as that of Proposition~\ref{prop:unifclopen} for existence of consistent (with respect to $\nll$ and $\alt$) ternary tests that control both the Type-I and Type-II error for arbitrary $\alpha$ and $\beta$, all without making any assumptions on $\world$.

%% file: math_ternary.tex
\section{Ternary Testing}\label{sec:three-arm}



As mentioned earlier, $\nll, \alt \subseteq \cP\Paren{\cX}$ where $\cX$ is a separable metric space and $\cP\Paren{\cX}$ denotes the set of all Borel probability measures. In general, for overlapping hypotheses $\nll$ and $\alt$, $\world \triangleq \nll \cup \alt$ is partitioned in three disjoint sets, $\prt_0 \triangleq \nll \setminus \alt, \prt_1 \triangleq \alt \setminus \nll$ and $\prt_{\idk} \triangleq \nll \cap \alt$. Table~\ref{tab:truth_all_error} summarizes an outcome-truth table for ternary hypothesis testing. We denote errors by the outcome-truth pairs, i.e., for example, if a test outputs $0$ on a sample drawn from the underlying distribution in $\prt_2$ then it makes a $\Paren{0,\idk}$-error. We will extend this notation to binary tests as well; i.e., a false positive is a $(1,0)$-error, and a false negative is a $(0,1)$-error. A `ternary test' $\test$ is a sequence of measurable functions $\Paren{\phi_1, \phi_2, \cdots}$ where $\phi_n : \cX^n \mapsto \left \{ 0,1,\idk \right \}$ with the subscript denoting the sample-size. We will abuse notation to refer to indexed families of ternary tests as ternary tests.

\subsection{Testability for Ternary Testing}

 Given a pair of overlapping hypotheses $\nll$ and $\alt$, we will now define desiderata that ternary testing must satisfy. For binary tests, there is a long line of work that specifies such desiderata \citep{Berger51, LeCamSchwartz60, DemboPeres94, Genin18,GeninKelly17,Ermakov17,BoekenSGM25} in terms of consistency, uniform consistency and false positive and false negative error-control guarantees. We define analogous notions of consistency and error-control for the ternary case. 
\begin{definition}[Weak and Strong Consistency]
Given hypotheses $\nll, \alt \subseteq \mathcal{P}\Paren{\mathcal{X}}$, a ternary test $\phi = \left \lbrace \test_n \right \rbrace_{n \in \NN}$ where $\phi_n : \mathcal{X}^n \mapsto \left \lbrace 0,1,\idk \right \rbrace$ is 
\begin{enumerate}
    \item weakly consistent if  for all $i \in \left \lbrace 0,1,\idk \right \rbrace$, for all $P \in R_i$, $ \lim\limits_{n \rightarrow \infty} P^n\Paren{\test_n\Paren{X^n}=i} = 1$ ;
    \item strongly consistent if for all $i \in \left \lbrace 0,1,\idk \right \rbrace$, for all $P \in R_i$, $ P^{\infty}\Paren{\test_n\Paren{X^n} \neq i \text{ for finitely many }n} = 1$.
\end{enumerate}
\end{definition}
 A pair of hypotheses $\nll, \alt$ is weakly or strongly consistent if there exists a test $\test$ that is weakly or strongly consistent, respectively. For $P \in \prt_i$, we term $P^n\Paren{\test_n\Paren{X^n}=i}$ as \textit{probability of correct detection} (PCD). Therefore, a weakly consistent test implies that PCD approaches $1$ in the limit of infinite sample-size for all $P \in \prt_i$, for all $i$. In the rest of the paper, we drop the distinction between weak and strong consistency since all our results hold for both cases. We extend the notions of error control to the ternary hypothesis testing setup below, except that errors are controlled asymptotically and not for every $n$. 

\begin{definition}[$(i,j)$-error control]\label{def:ij_EC}
    A ternary test $\test =\left \lbrace \test_{\alpha,n} \right \rbrace_{n\in \NN, \alpha > 0}$ is said to asymptotically control the $(i,j)$-error if for all $\alpha > 0$,
        $\lim\limits_{n\rightarrow \infty}\sup\limits_{P \in \prt_j } P^n\Paren{\phi_{\alpha,n}(X^n) = i } \leq \alpha$. 
\end{definition}
In other words, controlling the $(i,j)$-error is an error-guarantee on concluding $i$ when the data-generating distribution belongs to $\prt_j$. Note that if there exists a ternary test that controls the $(i,j)$-error then there exists a ternary test $\left \lbrace \test'_n \right \rbrace_{n \in \NN}$ such that $\lim\limits_{n\rightarrow \infty}\sup\limits_{P \in \prt_j } P^n\Paren{\test'_{n}(X^n) = i }  = 0$ (See Section~\ref{app:asymp_control_UC} for a proof).


\begin{definition}[Ternary-Testable Hypotheses ($\TT\Paren{I}$)]
  For $I \subseteq \{\prt_0, \prt_1, \prt_{\idk}\}$, $\Paren{\nll, \alt} \subseteq \cP\Paren{\cX}^2$ is $\TT\Paren{I}$ if there exists a consistent ternary test $\test =\left \lbrace \test_{\bar{\alpha},n} \right \rbrace_{n \in \NN, \bar{\alpha} \succ 0}$ such that $\test$ controls all $(i,j)$-errors for $i\neq j$, $\prt_j \in I$ where $\bar{\alpha} = \{ \alpha_{ij} \}_{i\neq j, \prt_j \in I}$ and $\bar{\alpha} \succ 0$ implies each component is positive. 
\end{definition}

In terms of Table~\ref{tab:truth_all_error}, a $\TT\Paren{I}$ pair of hypotheses is equivalent to the existence of a test that controls the error-rate of each of the off-diagonal cells in columns corresponding to sets in $I$. 


To derive equivalent topological conditions on $\nll, \alt$ that are $\TT\Paren{I}$, we reduce the existence of ternary tests to the existence of appropriate binary tests, which implies topological conditions on hypotheses from the results listed in Section~\ref{sec:prelim}.\ifdefined \SINGLE \else The proof can be found in Appendix~\ref{app:pf_nec}.\fi



\begin{restatable}{lemma}{nec}\label{lem:ttk_implies_kclosed}
    If $\Paren{\nll, \alt}$ is $\TT\Paren{I}$, then all $R \in I$ are closed and all $R \notin I$ are $F_{\sigma}$ in the weak topology on $\cP\Paren{\cX}$. 
\end{restatable}

\ifdefined\SINGLE 
\begin{proof}
Since $\Paren{\nll, \alt}$ are $\TT\Paren{I}$, for each $\prt_i \in I$, there exists a consistent binary test to test $\prt_i$ as the null against $\world \setminus \prt_i$ as the alternative that is uniformly consistent with respect to $\prt_i$. By Proposition~\ref{prop:unifclosed}, $\prt_i$ is closed in the weak topology on $\cP\Paren{\cX}$ for each $\prt_i \in I$. For $\prt_j \notin I$, there exists a consistent binary test to test  $\prt_j$ as the null against $\world \setminus \prt_j$ as the alternative implying, by Proposition~\ref{prop:consistencyfsigma}, that $\prt_j$ is $F_{\sigma}$ in the weak topology on $\cP\Paren{\cX}$. 
\end{proof}
\else \fi

In practice, the utility of these necessary topological conditions is to check whether it is possible to construct `good' ternary tests for the notion of `good' that the practitioner cares about. Concretely, if these topological conditions do not hold, then one can abandon the search for a good ternary test. However, if they hold, a natural question is whether these topological conditions can guide in constructing a ternary test with desired error-guarantees. We show how this can be done in the following section.

\section{Two-Stage Ternary Tests: Combining Binary Tests}\label{sec:two-stage}
A natural way to obtain ternary tests is by combining two or more binary tests. Indeed, given that many binary tests have been developed and analyzed for a wide range of hypotheses, we would like to leverage existing binary tests to construct ternary tests. We term any such combination of two binary tests a ``two-stage ternary test''. In this section, we focus on the following specific form of two-stage ternary tests. In the first stage, we test for one of the $\prt_i$'s and its complement in $\world$\footnote{All complements are taken with respect to $\world$.} using a binary test. In this stage, depending on whether $\prt_i$ or $\prt_i^{\complement}$ is the null hypothesis, we term the resulting ternary test `Split-alternative' (SA) and `Split-null' (SN), respectively. If in the first stage the binary test concludes to either not reject $\prt_i$ or reject $\prt_i^{\complement}$, the ternary test concludes in favor of $\prt_i$. Otherwise, we `split'  $\prt_i^{\complement}$ and perform another binary test where the null and alternative hypotheses are either of the constituent sets of $\prt_i^{\complement}$, i.e., $\prt_j, \prt_k$ such that $\prt_i^\complement = \prt_j \cup \prt_k$. See Figure~\ref{fig:SA} and Figure~\ref{fig:SN} \ifdefined\SINGLE \else in Appendix~\ref{app:error_rate}\fi for an illustration. We now relate the properties of the resulting ternary tests in terms of the properties of its constituent binary tests. 

Consider two binary tests $\test_{1}, \test_{2}$ where the former tests the null hypothesis $\prt_{1N}$ against its complement which we denote by $\prt_{1A}$, and the latter tests the null hypothesis $\prt_{2N}$ against the alternative $\prt_{2A}$. 
Table~\ref{tab:error_sa} and Table~\ref{tab:error_sn} summarize how error-rates for errors in the outcome-truth table of Table~\ref{tab:truth_all_error} relate to the false positive ($(1,0)$-error) and false negative ($(0,1)$-error) error rates of the constituent binary tests, $\test_{1}$ and $\test_{2}$ which we denote by $\fone_{1}, \fzero_{1}$, and $\fone_{2}, \fzero_{2}$, respectively, i.e.,
\begin{align}
    \lim\limits_{n \rightarrow \infty} \sup\limits_{P \in \prt_{1N}} P^n\Paren{\phi_{1,n}(X^n) = 1} &\leq \alpha_1, \label{eq:ecalpha1} \\
    \lim\limits_{n \rightarrow \infty}  \sup\limits_{P \in \prt_{1A}}  P^n\Paren{\phi_{1,n}(X^n)=0 } &\leq \beta_1, \label{eq:ecbeta1} \\
    \lim\limits_{n \rightarrow \infty} \sup\limits_{P \in \prt_{2N}} P^n\Paren{\phi_{2,n}(X^{n}) = 1 } &\leq \alpha_2, \label{eq:ecalpha2}\\
    \lim\limits_{n \rightarrow \infty} \sup\limits_{P \in \prt_{2A}} P^n\Paren{\phi_{2,n}(X^{n})=0} &\leq \beta_2 . \label{eq:ecbeta2}
\end{align}
\ifdefined\SINGLE
\else The derivation of one concrete instance can be found in Appendix~\ref{app:error_rate}.\fi
We also show that sample reuse does not affect the error-rate analysis. 

\ifdefined\SINGLE
\input{error_rate_analysis}
\else \fi

\begin{table*}
\caption{}\label{tab:global}
    \begin{subtable}[t]{\textwidth}
    \centering 
\begin{tabular}{|l|*{3}{c|}}
      \hline
      \diagbox[width=\dimexpr \textwidth/5+2\tabcolsep\relax, height=1.2cm]{\makecell[c]{\hspace{0em}\vspace{0.1em}{{\footnotesize{Outcome}}}}}{\makecell[c]{\vspace{-0.8em}{{\footnotesize{Truth}}}}}
        & $\prt_0 \Paren{\nll \setminus \alt}$ & $ \prt_1 \Paren{\alt  \setminus \nll}$ & $ \prt_{\idk} \Paren{\nll \cap \alt}$ \\
      \hline
      $0$ & True $0$ & $\Paren{0,1}$-error & $\Paren{0,\idk}$-error \\
      \hline
      $1$ & $\Paren{1,0}$-error & True $1$ & $\Paren{1,\idk}$-error \\
      \hline
      $\idk$ & $\Paren{\idk,0}$-error &$\Paren{\idk,1}$-error & True $2$ \\
      \hline
    \end{tabular}
             \caption{Truth-Outcome table for a general ternary test. All off-diagonal cells are considered errors.}\label{tab:truth_all_error}
    \end{subtable}
    \begin{subtable}[t]{0.45\textwidth}\centering
    \begin{tabular}{|l|*{3}{c|}}
      \hline
      \diagbox[width=\dimexpr \textwidth/5+2\tabcolsep\relax, height=1.2cm]{\makecell[c]{\hspace{0.2em}{\rotatebox{327}{\footnotesize{Outcome}}}}}{\makecell[c]{{\rotatebox{324}{\footnotesize{Truth}}}}}
        & $\prt_{1A}$ & $\prt_{2N}$ & $\prt_{2A}$ \\
      \hline
      $1A$ & - & $\fone_1$ & $\fone_1$ \\
      \hline
      $2N$ & $\fzero_1$ & - & $\fzero_2$\\
      \hline
      $2A$ & $\fzero_1$ & $\fone_2$ & - \\
      \hline
    \end{tabular}
    \caption{Error table for $\sn$ test.} \label{tab:error_sn}
    \end{subtable} 
    \begin{subtable}[t]{0.45\textwidth}\centering
    \begin{tabular}{|l|*{3}{c|}}
      \hline
      \diagbox[width=\dimexpr \textwidth/5+2\tabcolsep\relax, height=1.2cm]{\makecell[c]{\hspace{0.2em}{\rotatebox{327}{\footnotesize{Outcome}}}}}{\makecell[c]{{\rotatebox{324}{\footnotesize{Truth}}}}}
        & $\prt_{1N}$ & $\prt_{2N}$ & $\prt_{2A}$ \\
      \hline
      $1N$ & - & $\fzero_1$ & $\fzero_1$ \\
      \hline
      $2N$ & $\fone_1$ & - & $\fzero_2$\\
      \hline
      $2A$ & $\fone_1$ & $\fone_2$ & - \\
      \hline
    \end{tabular}
    \caption{Error table for $\sa$ test.} \label{tab:error_sa}
    \end{subtable} 
\end{table*}

\subsection{Completeness of Two-Stage Ternary Testing}

So far, we related error properties of two-stage ternary tests with error properties of constituent binary tests. In this section, we show that the necessary topological conditions for a pair of hypotheses to be $\TT\Paren{I}$, i.e., necessary topological conditions for the existence of a `good' ternary test, implies the existence of a good two-stage ternary test. Therefore, the topological conditions are both necessary and sufficient for the existence of a good ternary test. 




\begin{restatable}{lemma}{suff}\label{thm:ttk2k}
    Let $I \subseteq \{ \prt_0, \prt_1, \prt_2 \}$. If all $\prt \in I$ are closed and all $\prt \notin I$ are $F_{\sigma}$ in the weak topology on $\mathcal{P}\Paren{\cX}$, then there exists a consistent two-stage ternary test that controls all $(i,j)$-errors such that $i\neq j$ and $\prt_j \in I$.
\end{restatable}

\input{optimality}

%% file: error_rate_analysis.tex
\begin{figure*}[t]
\begin{minipage}{0.40\textwidth}
\raggedright
\begin{tikzpicture}[
    grow=right, 
    every node/.style={draw, rectangle, rounded corners, minimum width=1cm, fill=blue!5},
    edge from parent/.style={draw, -{Stealth}, thick},
    level 1/.style={sibling distance=2cm, level distance=2cm},
    level 2/.style={sibling distance=1.5cm, level distance=2.5cm},
    edge label/.style={draw=none, fill=white, font=\tiny, midway, rectangle, inner sep=2pt}
]
\node {$\phi_1$}
    child {
        node {Output $i$}
        edge from parent node[edge label] {reject $\prt_i$}
    }
    child {
        node {$\phi_2$}
        child {
            node {Output $j$}
            edge from parent node[edge label] {reject $R_j$}
        }
        child {
            node {Output $k$}
            edge from parent node[edge label] {not reject $R_j$}
        }
        edge from parent node[edge label] {not reject $\prt_i$}
    };
\end{tikzpicture} \caption{SN test} 
        \label{fig:SN}
\end{minipage}
 \hspace{5mm}
\begin{minipage}{0.40\textwidth}
\raggedright
\begin{tikzpicture}[
    grow=right, 
    every node/.style={draw, rectangle, rounded corners, minimum width=1cm, fill=blue!5},
    edge from parent/.style={draw, -{Stealth}, thick},
    level 1/.style={sibling distance=2cm, level distance=2cm},
    level 2/.style={sibling distance=1.5cm, level distance=2.5cm},
    edge label/.style={draw=none, fill=white, font=\tiny, midway, rectangle, inner sep=2pt}
]
\node {$\phi_1$}
child {
        node {$\phi_2$}
        child {
            node {Output $j$}
            edge from parent node[edge label] {reject $R_j$}
        }
        child {
            node {Output $k$}
            edge from parent node[edge label] {not reject $R_j$}
        }
        edge from parent node[edge label] {reject $\prt_i^{\complement}$}
    }
    child {
        node {Output $i$}
        edge from parent node[edge label] {not reject $\prt_i^{\complement}$}
    };
    
\end{tikzpicture}\caption{SA test}\label{fig:SA}
\end{minipage}
\end{figure*}

Consider a two-stage SA ternary test with constituent binary tests $\phi_1, \phi_2$
where  $\prt_{1N} = \prt_2$, $\prt_{2N} = \prt_0$. The combined ternary test is defined as
\begin{equation}\label{eq:test_inter_disj}
 \test_{n}(x^{n}) = 
    \begin{cases}
     \idk & \text{ if } \test_{1,n}(x^n)=0, \\
     0 & \text{ if } \test_{1,n}(x^n)=1, \test_{2,n}(x^{n})=0, \\
    1 & \text{ if } \test_{1,n}(x^n)=1, \test_{2,n}(x^{n})=1.
    \end{cases}
\end{equation}
For simplicity, we assume that both tests use the same number of samples. 
We now relate the error-rates of the two-stage ternary test to the asymptotic false positive and false-negative rates of its constituent binary tests. 

\textbf{ $\Paren{\idk,0}$-error and $\Paren{\idk,1}$-error}: The error rates for these errors are given by 
\begin{align*}
    &\sup\limits_{\param \in \prt_0} \param^{n}\Paren{\test_{n}\Paren{X^{n}}=\idk} = \sup\limits_{\param \in \prt_0} \param^n\Paren{\test_{1,n}\Paren{X^n}=0},\\
    &\sup\limits_{\param \in \prt_1} \param^{n}\Paren{\test_{n}\Paren{X^{n}}=\idk} = \sup\limits_{\param \in \prt_1} \param^n\Paren{\test_{1,n}\Paren{X^n}=0}
\end{align*}
respectively. From \eqref{eq:ecbeta1}, since $\prt_{1A} = \prt_0 \cup \prt_1$,
\begin{equation*}
    \max\Paren{\lim\limits_{n \rightarrow \infty} \sup\limits_{\param \in \prt_0}\param^{n}\Paren{\test_{n}\Paren{X^{n}}=\idk},\lim\limits_{n \rightarrow \infty} \sup\limits_{\param \in \prt_1}  \param^{n}\Paren{\test_{n}\Paren{X^{n}}=\idk}} \leq \fzero_1. 
\end{equation*}

\textbf{$\Paren{0,1}$-error and $\Paren{1,0}$-error}: The error rates for these errors are given by 
\begin{align*}
    \lim\limits_{n \rightarrow \infty}\sup\limits_{\param \in \prt_1} \param^{n}\Paren{\test_{n}\Paren{X^{n}}=0} &=  \lim\limits_{n \rightarrow \infty}\sup\limits_{\param \in \prt_1} \param^{n}\Paren{\test_{1,n}\Paren{X^n}=1 \cap \test_{2,n}\Paren{X^{n}}=0 }\\
    &\leq  \lim\limits_{n \rightarrow \infty}\sup\limits_{\param \in \prt_1} \param^n\Paren{\test_{2,n}\Paren{X^{n}}=0} \leq \fzero_2,
\end{align*}
\begin{align*}
     \lim\limits_{n \rightarrow \infty}\sup\limits_{\param \in \prt_0} \param^{n}\Paren{\test_{n}\Paren{X^{n}}=1} &=  \lim\limits_{n \rightarrow \infty}\sup\limits_{\param \in \prt_0} \param^{n}\Paren{\test_{1,n}\Paren{X^n}=1 \cap \test_{2,n}\Paren{X^{n}}=1} \\
    &\leq  \lim\limits_{n \rightarrow \infty}\sup\limits_{\param \in \prt_0} \param^n\Paren{ \test_{2,n}\Paren{X^{n}}=1} \leq \fone_2,
\end{align*}
 from \eqref{eq:ecbeta2} and \eqref{eq:ecalpha2} respectively. 

\textbf{$\Paren{0,\idk}$-error and $\Paren{1,\idk}$-error}: The error rates are given by
\begin{align*}
    \lim\limits_{n \rightarrow \infty}\sup\limits_{\param \in \prt_{\idk}} \param^{n}\Paren{\test_{n}\Paren{X^{n}}=0} &=  \lim\limits_{n \rightarrow \infty}\sup\limits_{\param \in \prt_{\idk}} \param^{n}\Paren{\test_{1,n}\Paren{X^n}=1 \cap \test_{2,n}\Paren{X^{n}}=0 },\\
     \lim\limits_{n \rightarrow \infty}\sup\limits_{\param \in \prt_{\idk}} \param^{n}\Paren{\test_{n}\Paren{X^{n}}=1} &=  \lim\limits_{n \rightarrow \infty}\sup\limits_{\param \in \prt_{\idk}} \param^{n}\Paren{\test_{1,n}\Paren{X^n}=1 \cap \test_{2,n}\Paren{X^{n}}=1}, 
\end{align*}
respectively. Since the right hand side of both the above equations can be bound by $\fone_1$ from \eqref{eq:ecalpha1}, 

\textbf{Remark}: Note that the proofs above reuse samples across the two stages. 

%% file: optimality.tex
\ifdefined\SINGLE \begin{proof}

    $I=\emptyset$: The statement is trivially true. 

    $|I|=1$: If $\prt_c \in I$ is closed in the weak topology on $\cP\Paren{\cX}$, by Proposition~\ref{prop:unifclosed}, there exists a consistent binary test that asymptotically controls the $(1,0)$-error-rate. Using this test in the first stage of an SA test where $\prt_c$ corresponds to $\prt_{1,N}$ controls the two required errors in Table~\ref{tab:error_sa}. For the second stage, since the remaining sets are $F_{\sigma}$, by Proposition~\ref{prop:consistencyfsigma}there exists a consistent binary test, thus making the resulting ternary test consistent with desired error control.  

    ${|I|=2}$ : Denote the closed sets by $\prt_{c_1}$ and $\prt_{c_2}$. Consider a SN test with $\prt_{1N}$ as $\prt_{c_1} \cup \prt_{c_2}$. Since $\prt_{c_1}$ and $\prt_{c_2}$ are closed in the weak topology on $\cP\Paren{\cX}$, their union is closed and since it is a subset of a compact set $\world$, it is also compact. Therefore, by Proposition~\ref{prop:unifclopen}, there exists a uniformly consistent test for testing $\prt_{c_1}$ as the null against $\prt_{c_2}$ as the alternative (or vice versa). This implies that in Table~\ref{tab:error_sn}, the errors corresponding to $\fone_1, \fone_2, \fzero_2$ are controlled. Since uniform consistency implies consistency, the existence of a consistent ternary test with desired error control is guaranteed. 

    $|I|=3$: Since all the sets are closed in the weak topology on $\cP\Paren{\cX}$, any pairwise union is closed. Since a closed subset of a compact set is compact, there exist constituent binary tests that are either a SA test or a SN test and asymptotically control both $(1,0)$-errors and $(0,1)$-errors, implying that all errors in Table~\ref{tab:error_sa} and Table~\ref{tab:error_sn} are controlled and the resulting ternary test is consistent. 
\end{proof}
 \else See Appendix~\ref{app:subsec:tt2k} for a proof.\fi From Lemma~\ref{lem:ttk_implies_kclosed} and Lemma~\ref{thm:ttk2k}, we conclude that two-stage ternary testing is complete for ternary testing where errors are controlled column-wise.
\begin{restatable}{theorem}{completeness}\label{thm:completeness_2stt}
    Let $I \subseteq \{R_0, R_1, R_2\}$, then the following are equivalent:
    \begin{enumerate}
        \item The sets in $I$ are closed in $\mathcal{P}(\mathcal{X})$, and the remaining are $F_\sigma$ in the weak topology on $\cP\Paren{\cX}$;
        \item There exists a consistent ternary test that controls all $(i, j)$-errors such that $i \neq j, \prt_j \in I$;
        \item There exists a consistent two-stage ternary test that controls all $(i, j)$-errors such that $i \neq j, \prt_j \in I$.
    \end{enumerate}
\end{restatable}
\begin{proof}
    $2 \implies 1 \implies 3$ from Lemma~\ref{lem:ttk_implies_kclosed} and \ref{thm:ttk2k}, respectively. $3 \implies 2$ by definition.
\end{proof}

\subsection{Guidelines for Constructing Two-Stage Ternary Tests}

There are $12$ possible two-stage ternary test options available to the practitioner, namely $6$ options for the null hypothesis of the first stage and for each such option, $2$ options for the second stage. Therefore, constructing a two-stage ternary test that has desired error control properties requires choosing one of the $12$ options and then constructing constituent binary tests. While we do not address the latter\footnote{We refer the reader to the vast literature on (binary) hypothesis testing.}, 
in this section, we use the topological conditions of Theorem~\ref{thm:completeness_2stt} as a guide to the former part.

When the topological conditions in Theorem~\ref{thm:completeness_2stt} are satisfied for a particular $I$, we consider different cases depending on how many sets in $\{ \prt_0, \prt_1, \prt_2 \}$ are closed in the weak topology on $\mathcal{P}\Paren{\cX}$. Figure~\ref{fig:dt} in Appendix~\ref{app:dt} illustrates decision trees that provide a guide to which two-stage ternary tests provide the desired error control under different topological properties of $\prt_0, \prt_1, \prt_2$. Note that while all the tests in the leaf nodes provide the desired error control, those suggested in blue control additional errors. While the tests highlighted in blue are suggested because they control the most number of errors, it must be noted that this excludes tests where $\prt_{2N} \cup \prt_{2A}$ is not compact because Proposition~\ref{prop:unifclosed} depends on such a compactness assumption. As we shall see in Section~\ref{sec:IV}, Proposition~\ref{prop:unifclosed_ermakov} can be used to guarantee existence of uniformly consistent binary tests even when $\prt_{2N} \cup \prt_{2A}$ is not compact. Irrespective of topological properties, for any of the aforementioned $12$ options, if there exist consistent binary tests for the pairs of hypotheses in the option, with error rates $\alpha_1,\beta_1,\alpha_2,\beta_2$, then the error rates of the resulting SA or SN ternary test are given in Table~\ref{tab:error_sa} and Table~\ref{tab:error_sn}, respectively.

%% file: iv_three_arm.tex
\section{Applications to Testing Causal Queries}\label{sec:applications}
\subsection{Testing the Instrumental Variable (IV) Inequalities}\label{sec:IV}
\begin{figure}[t]
\begin{minipage}{0.45\textwidth}
     \centering
            \begin{tikzpicture}
            \tikzstyle{vertex}=[circle,fill=none,draw=black,minimum size=17pt,inner sep=0pt]
\node[vertex] (Z) at (0,0) {$Z$};
\node[vertex] (Y) at (3,0) {$Y$};
\node[vertex] (X) at (1.5,0) {$X$};
\path (Z) edge (X);
\path (X) edge (Y);
\path[bidirected] (X) edge[bend left=60] (Y);
            \end{tikzpicture}
        \caption{Causal graph of $M \in \modelivrelax$} 
        \label{fig:iv}
\end{minipage}  
\begin{minipage}{0.45\textwidth}
     \centering
            \begin{tikzpicture}
            \tikzstyle{vertex}=[circle,fill=none,draw=black,minimum size=17pt,inner sep=0pt]
\node[vertex] (Z) at (0,0) {$Z$};
\node[vertex] (Y) at (3,0) {$Y$};
\node[vertex] (X) at (1.5,1) {$X$};
\path (Z) edge (X);
\path (X) edge (Y);
\path (Z) edge (Y);
\path[bidirected] (X) edge[bend left=60] (Y);
            \end{tikzpicture}
        \caption{Causal graph of $\model \in \modelsivedge$} 
        \label{fig:iv_gen}
        \end{minipage}
        \end{figure}

We will now apply the ternary testing framework to provide a test for the IV inequalities \citep{Pearl95}. Define the IV model class, denoted by $\modelivrelax$, to be the set of causal models whose causal graph is a subgraph of Figure~\ref{fig:iv}. Define $\modelsivedge$ to be the set of causal models whose causal graph is a subgraph of Figure~\ref{fig:iv_gen} and note that $\modelivrelax \subset \modelsivedge$. A Markov kernel $K(X,Y \mid Z)$ defined on discrete random variables $X,Y,Z$ is said to satisfy the IV inequalities if 
\begin{equation}\label{eq:iv_notes}
    \max_{x \in \cX} \sum\limits_{y \in \cY} \max_{z \in \cZ} K(X,Y \mid Z) \leq 1.
\end{equation}
\citet{Pearl95} proved that if $X,Y,Z$ are discrete random variables, then a necessary condition for $\model \in \modelivrelax$ is that $P_{\model}\Paren{X,Y \mid \doop{Z}}$ satisfies \eqref{eq:iv_notes}\footnote{Note that \citet{Pearl95} and most of subsequent literature frames the IV inequalities as conditions that the observational distribution $P_{\model}(X,Y \mid Z)$ satisfies. However, this only holds under the tacit assumption of positivity which is taken for granted since the instrument, $Z$, is often randomized. When positivity of $Z$ is not guaranteed, IV inequalities are correctly expressed as conditions only on $P_{\model}(X,Y \mid \doop{Z})$ and thus, identifiability is not guaranteed.}. We view testing the IV inequalities as testing the binary causal query that evaluates to $0$ if $P_{\model}\Paren{X,Y \mid \doop{Z}}$ satisfies \eqref{eq:iv_notes}. We only consider the case where the instrument is binary since it is known that the IV inequalities are, in general, not sufficient for falsifying the IV model class for non-binary instruments \citet{Bonet01, SongGCR25}. Further, \citet{BhadaneMBZ25} show that for the binary instrument, binary outcome, categorical treatment case, the IV inequalities are sufficient. 

 Testing the IV inequalities is the same as testing the causal null hypothesis,
\ifdefined\SINGLE
\begin{equation*}\label{eq:iv_null}
    \nulliv \triangleq \left \{ \model \in \modelsivedge: P_{\model}\Paren{X,Y \mid \doop{Z}} \text{ satisfies } \eqref{eq:iv_notes} \right \}.
\end{equation*}
\else 
\begin{equation*}\label{eq:iv_null}
    \nulliv \triangleq \left \{ \model \in \modelsivedge: P_{\model}\Paren{X,Y \mid \doop{Z}} \text{ satisfies } \eqref{eq:iv_notes} \right \}.
\end{equation*}
\fi 

$\nulliv$ induces the set of observational distributions, $\distivnull \triangleq \left \{ P_{\model}: \model \in \nulliv \right \}$ and the causal alternative hypothesis, $\altiv \triangleq \modelsivedge \setminus \nulliv$ induces the set of observational distributions, $\distivalt \triangleq \left \{ P_{\model} : \model \in \altiv \right \}$.
As before, denote $\prt_0 = \distivnull \setminus \distivalt, \prt_1 = \distivalt \setminus \distivnull, \prt_{\idk} = \distivnull \cap \distivalt$. Further, denote $\world \triangleq \distivnull \cup \distivalt$ and note that for $\cX, \cY, \cZ$ being discrete and finite sets, $\world=\Delta$ where $\Delta$ is the probability simplex of dimension $|\cX|\times |\cY| \times 2$. 

We first characterize $\prt_2$ by showing that the
distributions in $\prt_2$ are the same as those that violate positivity of $Z$.
\ifdefined\SINGLE
\begin{lemma}
\begin{equation}\label{lem:positivity_intersection}
   \prt_2 = \{ P_{\model} : \cZ_{pos, \model} \neq \cZ,  \model \in \modelsivedge \}.
\end{equation}
\end{lemma}
\begin{proof}
For $\model \in \modelsivedge$, $P_{\model}\Paren{X,Y \mid \doop{Z=z}} = P_{\model}\Paren{X,Y \mid Z=z}$ if $P_{\model}\Paren{Z=z}>0$ for all $z \in \cZ$. Therefore, if $P_{M}(Z=z) > 0$ for all $z \in \cZ$, then $P_{\model}\Paren{X,Y \mid \doop{Z}}$ satisfies \eqref{eq:iv_notes} if and only if $P_{\model}\Paren{X,Y \mid Z}$ satisfies \eqref{eq:iv_notes} and vice-versa. This implies, 
$$\Paren{\distivnull \cap \distivalt} \cap \left \{ P_{\model} : P_{\model}(Z=z) > 0 \text{ for all } z \in \cZ, \model \in \modelsivedge \right \} = \emptyset.$$
This implies $\Paren{\distivnull \cap \distivalt} \subseteq \left \{ P_{\model} : P_{\model}(Z=z) = 0 \text{ for some } z \in \cZ, \model \in \modelsivedge \right \}$. Since, for any $P_{\model}$ such that $P_{\model}(Z=z')=0$ for some $z' \in \cZ, \model \in \modelsivedge$, there exists two models $\model_1$ and $\model_2$ with parameters the same as $\model$ except that $P_{\model_1}\Paren{X,Y \mid \doop{Z=z'}}$ satisfies \eqref{eq:iv_notes} and $P_{\model_2}\Paren{X,Y \mid \doop{Z=z'}}$ does not satisfy \eqref{eq:iv_notes}. Since $P_{\model} = P_{\model_1} \in \distivnull, P_{\model} = P_{\model_2} \in \distivalt, P_{\model} \in \distivnull \cap \distivalt$, concluding the proof. \end{proof}
\else 
\begin{restatable}{lemma}{positivityintersection}\label{lem:positivity_intersection}
\begin{equation*}
    \prt_2 = \{ P_{\model} : P_{\model}\Paren{z} = 0 \text{ for some } z \in \cZ, \model \in \modelsivedge \}.
\end{equation*}
\end{restatable}
The proof can be found in Appendix~\ref{app:subsec:positivity_intersection}.
\fi 
We outline topological properties of $\prt_0, \prt_1, \prt_2$ and $\world$ that will eventually lead us to a two-stage ternary test for the IV inequalities. 

$\textbf{\world}$: Since $\world$ can be regarded as a subset of $\RR^{|\cX| \times |\cY| \times |\cZ|}$, the topology of setwise convergence is equivalent to the product topology on $\RR^{|\cX| \times |\cY| \times |\cZ|}$. Since $\Delta$ is closed and bounded it is compact in the product topology and therefore, compact in the topology of setwise convergence. 

$\textbf{\prt}_2$: $\prt_2$ is closed in the weak topology on $\Delta$ since $ \exists z \in \cZ$ s.t. $P_{\model}(Z=z) = 0$.

$\textbf{\prt}_1$: For $\model \in \modelsivedge$, $P_{\model} \in \prt_1$ if and only if $P_{\model}\Paren{Z=z}>0$ for all $z \in \cZ$ and equation~\ref{eq:iv_notes} does not hold for $P_{\model}\Paren{X,Y \mid \doop{Z}} = P_{\model}\Paren{X,Y \mid Z}$. Since the mapping $P_{\model}\Paren{X,Y,Z} \mapsto P_{\model}\Paren{X,Y \mid Z}$ is continuous, $\prt_1$ is open in the weak topology on $\Delta$.



 

\begin{algorithm}[ht]
\caption{Two-stage Ternary Test for \eqref{eq:iv_notes} }\label{alg:iv}
\KwData{$(X_i,Y_i,Z_i)_{i=1}^{n}$}

\KwResult{$\test_n\Paren{(X_i,Y_i,Z_i)_{i=1}^{n}} \in \{0,1,2\}$}
 Define $f_{\text{pos}} : \cX^n \times \cY^n \times \cZ^n \mapsto \{0,1\}$ such that $f_{\text{pos}} ((X_i,Y_i,Z_i)_{i=1}^{n})=0$ if there exists $z \in \cZ$ such that $Z_i \neq z$ for all  $1\leq i \leq n$ and $1$ otherwise\;  
 Test $\prt_2$ vs $\prt_1 \cup \prt_0$ using $\test_1$ where $\test_1(X_i,Y_i,Z_i)_{i=1}^{n}=f_{\text{pos}} ((X_i,Y_i,Z_i)_{i=1}^{n})$ \;
  \eIf{$\test_1$ outputs $1$}{
   Test $\prt_0$ vs $\prt_1$ using IV test implemented by \citet{WangRR17}. $\phi_n$ outputs $1$ if IV test rejects null and $0$ otherwise\;}
  {
   $\test_n$ outputs $2$\;
  }
 
\end{algorithm}

\begin{algorithm}[t]
\caption{Two-stage Ternary Test for Treatment Efficacy Comparison }\label{alg:tec}
\KwData{$(X_i,Y_i)_{i=1}^{n}$, $c_{\text{threshold}}$}

\KwResult{$\test_n\Paren{(X_i,Y_i)_{i=1}^{n},c_{\text{threshold}}} \in \{0,1,2\}$}
 Test $\prt_0 \cup \prt_2 $ vs $\prt_1$ using a one-sided binomial test $\test_1$ for $P(X=1,Y=0)$ with threshold $1-c_{\text{threshold}}$ where $\test_1$ outputs $1$ if the null is rejected and $0$ otherwise\;
  \eIf{$\test_1$ outputs $0$}{
   Test $\prt_0 $ vs $\prt_2$ using a one-sided binomial test $\test_2$ for $P(X=1,Y=1)$ with threshold $c_{\text{threshold}}$\; $\test_n$ outputs $2$ if $\test_2$ rejects null and $0$ otherwise\;}
  {
   $\test_n$ outputs $1$\;
  }
 
\end{algorithm}

From Figure~\ref{fig:dt_1}, the above topological conditions suggest an SN test where $\prt_{1,A}$ is $\prt_1$ and $\prt_{2,N}$ is $\prt_2$ that controls the errors corresponding to $\fone_1, \fone_2$ in Table~\ref{tab:error_sn}, i.e., errors $(1,0), (0,2)$ and $(1,2)$ are controlled. A natural candidate for testing the IV inequalities, i.e., a positivity test followed by an IV inequalities test that assumes positivity, is given by an SA test with $\prt_{1N} = \prt_2, \prt_{2N} = \prt_0$. Since $\prt_{2N} \cup \prt_{2A}$ is not compact, Proposition~\ref{prop:unifclosed} does not apply. However, we use \cite{Ermakov17}'s Theorem 4.3 (Proposition~\ref{prop:unifclosed_ermakov}) to guarantee the existence of a test that asymptotically controls the $(1,0)$ error thus implying that the SA test and the aforementioned SN test control the same errors. Algorithm~\ref{alg:iv} specifies a SA test for the IV inequalities which uses a positivity test based on the empirical distribution and the IV inequalities test assuming positivity from \citet{WangRR17}. 

%% file: manski.tex
\subsection{Treatment Efficacy Comparison}
We now consider the problem of comparing the efficacy of a treatment against a threshold given data from an observational study. The treatment variable, $X$,
 has two possible states: untreated $(0)$, treated $(1)$. The outcome variable $Y$
 is binary indicating whether the treatment was effective $(1)$ or not $(0)$. Given a threshold $c \in [0,1]$, which might have been obtained, say from a large randomized control trial of a different treatment, and given observational data, we want to test whether the treatment efficacy is at least the threshold, i.e., whether $P\Paren{Y=1 \mid \doop{X=1}} \geq c$.

Define $\mathbb{M}_{TEC}$ to be the set of all acyclic causal models with endogenous variables $X$ and $Y$ where $\cX =\cY=\{0,1\}$. The causal null hypothesis is then defined as 
\begin{equation*}\label{eq:manski_null_causal}
    \nll_{TEC} \triangleq \{\model \in \mathbb{M}_{TEC} : P_{\model}\Paren{Y=1 \mid \doop{X=1}} \geq c\},
\end{equation*}
with the causal alternative hypothesis $\alt_{TEC} \triangleq \MM_{TEC} \setminus \nll_{TEC}$. The Manski bounds \citep{Manski90} can be used to bound the above query
\ifdefined\SINGLE
\begin{equation}
   P_{\model}\Paren{Y=1,X=1} + P_{M}\Paren{X \neq 1}
    \geq P_{\model}\Paren{Y=1 \mid \doop{X=1}} \geq P_{\model}\Paren{Y=1,X=1}. 
\end{equation}\label{eq:manski}  
\else 
\begin{align}
    &P_{\model}\Paren{Y=1,X=1} + P_{M}\Paren{X \neq 1} \nonumber \\
    &\geq P_{\model}\Paren{Y=1 \mid \doop{X=1}} \geq P_{\model}\Paren{Y=1,X=1}. \label{eq:manski}
\end{align}\fi
Since the Manski bounds are tight, the causal null hypothesis induces the set of observational distributions $\cP_{H^0_{TEC}} \triangleq \{ P_{\model} : \model \in \nll_{TEC} \}$ where
\begin{align*}
   \cP_{H^0_{TEC}} &= \{ P_{\model} : P_{\model}\Paren{Y=1,X=1} + P_{M}\Paren{X \neq 1} \geq c \}, \\
   &= \{ P_{\model} : P_{\model}\Paren{Y=0,X=1} \leq 1-c \}
\end{align*}
The causal alternative hypothesis induces the set of observational distributions $\cP_{H^1_{TEC}} \triangleq \{ P_{\model} : \model \in \alt_{TEC} \}$
\begin{equation*}
    \cP_{H^1_{TEC}} = \{ P_{\model} : P_{\model}\Paren{Y=1,X=1}  < c \}
\end{equation*}
As before, denote $\prt_0 = \cP_{\nll_{TEC}} \setminus \cP_{\alt_{TEC}}, \prt_1 = \cP_{\alt_{TEC}} \setminus \cP_{\nll_{TEC}}, \prt_2 = \cP_{\nll_{TEC}} \cap \cP_{\alt_{TEC}}$. Let $\world = \cP_{\nll_{TEC}} \cup \cP_{\alt_{TEC}}$, and $\Delta$ be the probability simplex of dimension $|\cX| \times |\cY|$. Note that $\prt_0 = \Delta \setminus \cP_{\alt_{TEC}}$ and $\prt_1 = \Delta \setminus \cP_{\nll_{TEC}}$ which implies that $\world = \Delta$. Since $\world$ can be regarded as a subset of $\RR^{|\cX| \times |\cY|}$, the topology of setwise convergence is equivalent to the product topology on $\RR^{|\cX| \times |\cY|}$ which implies that $\Delta$ is compact in the topology of setwise convergence. Further, note that $\prt_0$ is closed and $\prt_1$ is open in the weak topology on $\Delta$. By Figure~\ref{fig:dt_1}, a SN test with $\prt_{1A} = \prt_1$ and $\prt_{2N} = \prt_0$ provides maximum error control. Algorithm~\ref{alg:tec} specifies a SN test for treatment efficacy comparison where both tests use one-sided binomial proportions tests and control the $(1,0), (2,0)$ and $(1,2)$ errors.

%% file: numerical.tex
\subsection{Numerical Experiments}\label{subsec:numeric}
\begin{figure}[t]
\centering
\includegraphics[scale=0.55]{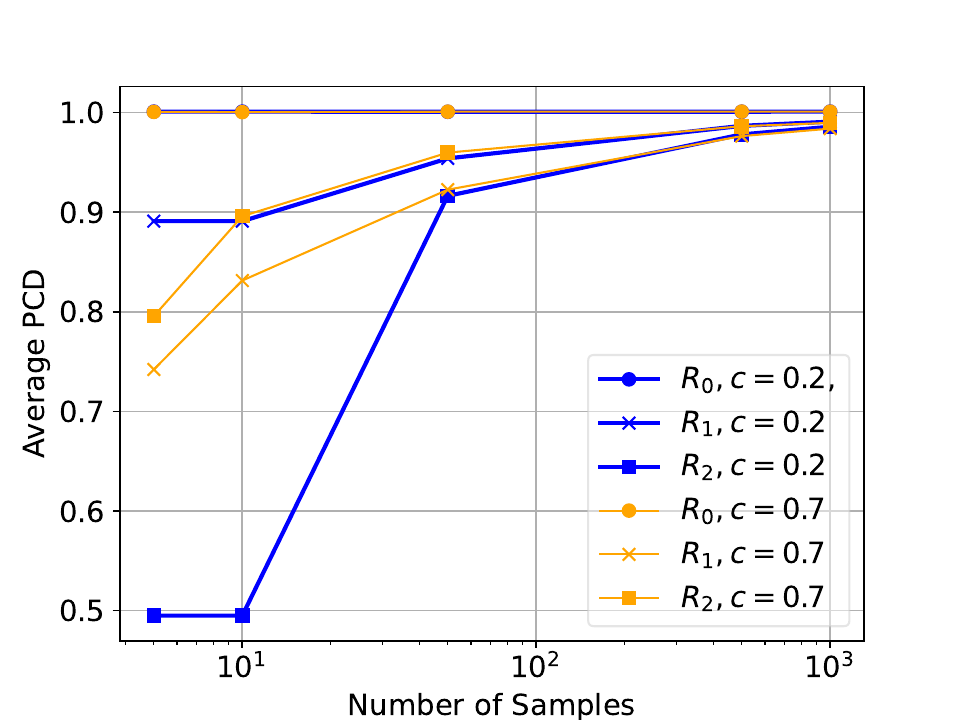}\caption{Average PCD of Algorithm~\ref{alg:tec} plotted as a function of number of samples. The size for the constituent binary tests was fixed at $0.025$ each. }\label{fig:power}
\end{figure}
We perform a simulation-based analysis for the treatment efficacy comparison case. To make the least parametric assumptions, we consider a data-generating process where a discrete random variable $U$ takes $|\cX| \times |\cY|^{|\cX|} = 8$ values such that each value corresponds to a response-function pair $(u_x,u_y)$ \citep{BalkePearl97} where $u_x \in \{0,1\}$ and $u_y$ is a one of four possible functions that map $\cX$ to $\cY$ . We sample $U \sim P_U$ and generate data $(X,Y) = (U_x,U_y(U_x))$. We sample $N_{\text{dist}}=1000$ distributions $P_U \sim \text{Dir}\Paren{1/8}$ and further sample $n$ i.i.d. samples of $(X,Y)$. Figure~\ref{fig:power} plots the average PCD of the ternary test for each of $\prt_0, \prt_1$ and $\prt_2$ as a function of $n$ for two values of $c$, which is akin to a binary hypothesis test power analysis. Note that for $\prt_0$, for both values of $c$, the power is high since the SN test of Algorithm~\ref{alg:tec} controls the $(1,0), (2,0)$ and $(1,2)$ errors.  

\ifdefined\SINGLE 
We compare the following natural, yet naive, ternary test that a practitioner might implement: Compute confidence intervals for the Manski bounds. If the threshold $c$ is included between the lower confidence bound of the lower Manski bound and the upper confidence bound of the upper Manski bound then the test returns $2$. If the threshold is lower than the the lower confidence bound of the lower Manski bound then the test returns $1$ and otherwise $0$. We compare our proposed ternary test of Algorithm~\ref{alg:tec} with this naive ternary test in Figure~\ref{fig:naive}. 

\begin{figure}[t]
\centering
\includegraphics[scale=0.55]{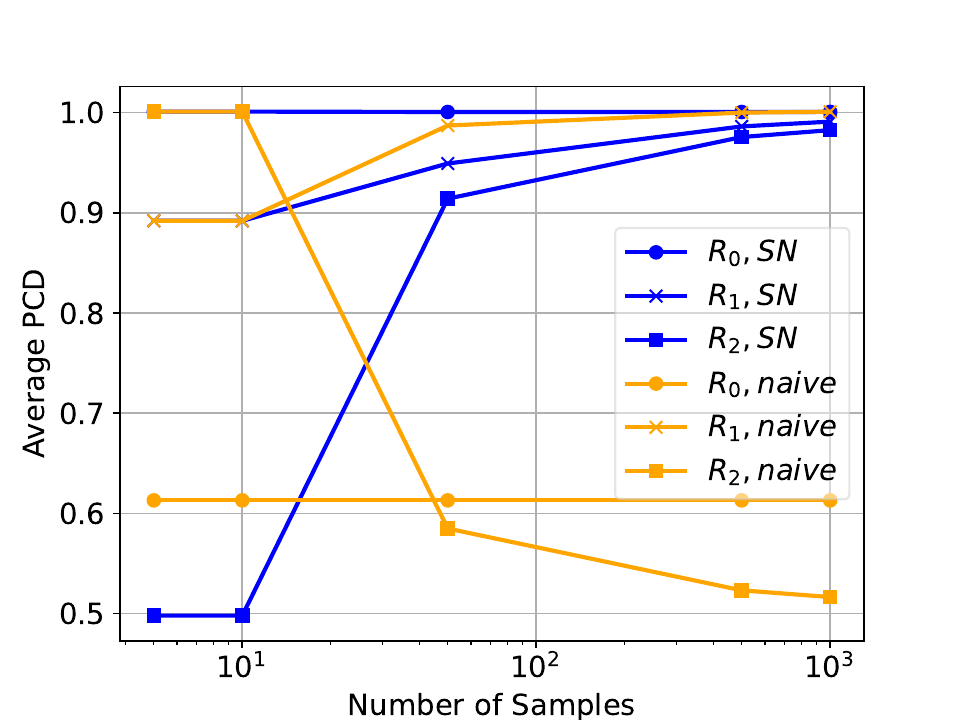}\caption{Comparison of the Average PCD of the naive ternary test versus Algorithm~\ref{alg:tec} plotted as a function of number of samples. The size for the constituent binary tests was fixed at $0.025$ each and the naive ternary test uses $97.5 \%$ confidence intervals for the Manski bounds.}\label{fig:naive}
\end{figure}
\else \fi

For the IV inequalities, we use the 1973 University of California, Berkeley graduate school admissions dataset \texttt{UCBAdmissions} \citet{UCBAdmissions} that contains counts for sex-department-admissions outcome for the $6$ largest departments at University of California, Berkeley. We follow the observation of \citet{BhadaneMBZ25}  that models this decision-making process as an instrumental variable model where sex is treated as an instrument, department choice as treatment and admission decision as outcome. On the Berkeley dataset, the positivity test outputs a p-value of $0$ since there is data on both sexes. The IV inequality test of \citet{WangRR17} outputs $0$
 as the p-value of the IV inequalities test 
 (i.e., less than the smallest positive number representable using double precision floating point format, i.e. $< 5 \times 10^{-324}$, 
 as reported in \citet{BhadaneMBZ25}. Thus, the IV inequalities are satisfied for the Berkeley dataset. 

%% file: discussion.tex
\section{Discussion}

We consider ternary testing for overlapping hypotheses which includes establishing desiderata related to error-guarantees of a ternary test in terms of uniform consistency. The existence of ternary tests that satisfy these desiderata can be characterized by topological conditions on the hypotheses. We introduced two-stage ternary tests that take advantage of the large number of binary tests developed so far. Since the design space of a ternary test is arguably larger than that of binary tests, the characterizing topological conditions guide our construction of two-stage ternary tests. We demonstrate this guide on two applications, namely testing IV inequalities and treatment efficacy comparisons.

We leave a) the relaxation of the compactness asssumption, and b) making the categorization of ternary-testable hypotheses finer by considering existence of ternary tests that control the error-rate for any combination of cells in the error table instead of errors in columns corresponding to sets in $I$, for future work.

%% file: appendix.tex
\section{Preliminaries (Causal Models)}\label{app:prelim_causal}
We outline a few definitions related to causal models that follow the formal setup of \cite{BongersFPM21}. 

\begin{definition}[Structural Causal Model (SCM)]
A \textbf{Structural Causal Model (SCM)} is a tuple $\model = \Paren{\enop,\exrv,\spc,f,P}$ where a) $\enop,\exrv$ are disjoint, finite index sets of \textbf{endogenous} and \textbf{exogenous} random variables respectively, b) $\spc = \prod_{i \in \enop \cup \exrv} \spc_i $ is the \textbf{domain} which is a product of standard measurable spaces $\spc_i$, c) for every $v \in \enop$, $f_v:\spc \mapsto \spc_{v}$ is a measurable function,  and $f = (f_v)_{v \in V}$ is called the \textbf{causal mechanism}; the equations $X_v = f_v(X)$ for every $v \in \enop$ are called the \textbf{structural equations}, and d) $P\Paren{\spc_{\exrv}} = \bigotimes_{w \in \exrv} P(X_w)$ is the \textbf{exogenous distribution} which is a product of probability distributions $P(X_w)$ on $\spc_{w}$.  
\end{definition}

\begin{definition}[Parent]
Let $\model = \Paren{\enop,\exrv,\spc,f,P}$ be an SCM. $k \in \enop \cup \exrv$ is a \textbf{parent} of $v \in \enop$ if and only if it is not the case that for all $x_{\enop \backslash k}, f_v(x_V,X_W)$ is constant in $x_k$ (if $k \in \enop$, resp. $X_k$ if $k \in \exrv$)  $P(\spc_{\exrv})$-a.s..
\end{definition}
\begin{definition}[Causal Graph]
             Let $\model = \Paren{\enop, \exrv, \spc, f,P}$ be an SCM. The \textbf{causal graph}, $\cg{\model}$, is a directed mixed graph with nodes $V$, directed edges $u \longrightarrow v$ if and only if $u \in \enop$ is a parent of $v \in \enop$, and bidirected edges $u \leftrightarrow v$ if and only if $ \,\exists w \in \exrv$ that is a parent of both $u,v \in \enop$. 
\end{definition}
For simplicity of exposition, we restrict attention to acyclic SCMs, i.e.\ SCMs whose causal graph is acyclic (contains no directed cycle $X \rightarrow \cdots \rightarrow Y \rightarrow X$).
\begin{definition}[Observational Distribution]
Given an acyclic SCM, $\model= \Paren{\enop,\exrv,\spc,f,P}$, the exogenous distribution, $P$ and the causal mechanism $f$ induce a probability distribution over the endogenous variables which is called the \textbf{observational distribution}, $P_{\model}(X_{\enop})$.
\end{definition}

\begin{definition}[Hard Intervention]
Given an acyclic SCM, $\model= \Paren{\enop,\exrv,\spc,f,P}$, an intervention target $\,T \subseteq \enop$, and an intervention value $x_T \in \cX_{T}$, the \textbf{intervened SCM} is defined as $\model_{\doop{X_T=x_T}} \triangleq \Paren{\enop,\exrv,\spc,(f_{V \backslash T},x_T),P}.$ Further, the observational distribution of the intervened SCM, $P_{\model_{\doop{X_T=x_T}}}$, is called an \textbf{interventional distribution}, and denoted by $P_{M}\Paren{X_V \mid \doop{X_T=x_T}}$.
\end{definition}

\section{Asymptotic Control and Uniform Consistency}\label{app:asymp_control_UC}
\begin{proposition}\label{prop:asymp_control_UC}
    Given a pair of hypotheses $\nll, \alt \subseteq \mathcal{P}\Paren{\cX}$, if there exists a ternary test that asymptotically controls the $(i,j)$-error then there exists a $\test' = \{ \test'_n \}_{n \in \NN}$ such that $\lim\limits_{n \rightarrow \infty} \sup\limits_{P \in \prt_j} P^{n}\Paren{\test'_n\Paren{X^n}=i} = 0$.
\end{proposition}

\begin{proof}
     By Definition~\ref{def:ij_EC}, there exists a ternary test $\test = \{\test_{\alpha,n}\}_{n \in \NN, \alpha >0}$ such that for all $\alpha >0, \lim\limits_{n \rightarrow \infty} \sup\limits_{P \in \prt_j} P^{n}\Paren{\test_{\alpha,n}\Paren{X^n}=i} \leq \alpha$. Consider the sequence $\alpha_m = \{ \frac{1}{m}\}_{m \in \NN}$. There exists $N_m \in \NN$ such that for all $n > N_m$, $\sup\limits_{P \in \prt_j} P^{n}\Paren{\test_{\frac{1}{m},n}\Paren{X^n}=i} \leq \frac{1}{m}$. Construct $\test'_n = \test_{\frac{1}{m},n}$ for $N_m < n \leq N_{m+1}$. Then for any $\varepsilon > 0$, there exists $m' \in \NN$ such that $\frac{1}{m'+1} < \varepsilon \leq \frac{1}{m'}$ and consequently there exists $N_{\varepsilon} = N_{m'+1}$ such that for all $n > N_{\varepsilon} = N_{m'+1}$, $$\sup\limits_{P \in \prt_j} P^{n}\Paren{\test'_{n}\Paren{X^n}=i} = \sup\limits_{P \in \prt_j} P^{n}\Paren{\test_{\frac{1}{m'+1},n}\Paren{X^n}=i} \leq \frac{1}{m'+1} < \varepsilon.$$ This implies that $\lim\limits_{n \rightarrow \infty} \sup\limits_{P \in \prt_j} P^{n}\Paren{\test'_n\Paren{X^n}=i} = 0$.
    \end{proof}

\ifdefined\SINGLE 
\else
\section{Error-Rate Analysis for Two-stage Ternary Test}\label{app:error_rate}

\input{error_rate_analysis}
\fi

\ifdefined\SINGLE
\else 
\section{Proofs}\label{app:pf_iv}
\subsection{Proof of Lemma~\ref{lem:ttk_implies_kclosed}}\label{app:pf_nec}
\nec*
\begin{proof}
Since $\Paren{\nll, \alt}$ are $\TT\Paren{I}$, for each $\prt_i \in I$, there exists a consistent binary test to test $\prt_i$ as the null against $\world \setminus \prt_i$ as the alternative that is uniformly consistent with respect to $\prt_i$. By Proposition~\ref{prop:unifclosed}, $\prt_i$ is closed in the weak topology on $\cP\Paren{\cX}$ for each $\prt_i \in I$. For $\prt_j \notin I$, there exists a consistent binary test to test  $\prt_j$ as the null against $\world \setminus \prt_j$ as the alternative implying, by Proposition~\ref{prop:consistencyfsigma}, that $\prt_j$ is $F_{\sigma}$ in the weak topology on $\cP\Paren{\cX}$. 
\end{proof}
 \subsection{Proof of Lemma~\ref{thm:ttk2k}}\label{app:subsec:tt2k}
 \suff* 
\begin{proof}

    $I=\emptyset$: The statement is trivially true. 

    $|I|=1$: If $\prt_c \in I$ is closed in the weak topology on $\cP\Paren{\cX}$, by Proposition~\ref{prop:unifclosed}, there exists a consistent binary test that asymptotically controls the $(1,0)$-error-rate. Using this test in the first stage of an SA test where $\prt_c$ corresponds to $\prt_{1,N}$ controls the two required errors in Table~\ref{tab:error_sa}. For the second stage, since the remaining sets are $F_{\sigma}$, by Proposition~\ref{prop:consistencyfsigma}there exists a consistent binary test, thus making the resulting ternary test consistent with desired error control.  

    ${|I|=2}$ : Denote the closed sets by $\prt_{c_1}$ and $\prt_{c_2}$. Consider a SN test with $\prt_{1N}$ as $\prt_{c_1} \cup \prt_{c_2}$. Since $\prt_{c_1}$ and $\prt_{c_2}$ are closed in the weak topology on $\cP\Paren{\cX}$, their union is closed and since it is a subset of a compact set $\world$, it is also compact. Therefore, by Proposition~\ref{prop:unifclopen}, there exists a uniformly consistent test for testing $\prt_{c_1}$ as the null against $\prt_{c_2}$ as the alternative (or vice versa). This implies that in Table~\ref{tab:error_sn}, the errors corresponding to $\fone_1, \fone_2, \fzero_2$ are controlled. Since uniform consistency implies consistency, the existence of a consistent ternary test with desired error control is guaranteed. 

    $|I|=3$: Since all the sets are closed in the weak topology on $\cP\Paren{\cX}$, any pairwise union is closed. Since a closed subset of a compact set is compact, there exist constituent binary tests that are either a SA test or a SN test and asymptotically control both $(1,0)$-errors and $(0,1)$-errors, implying that all errors in Table~\ref{tab:error_sa} and Table~\ref{tab:error_sn} are controlled and the resulting ternary test is consistent. 
\end{proof}

\subsection{Proof of Lemma~\ref{lem:positivity_intersection}}\label{app:subsec:positivity_intersection}
\positivityintersection*
\begin{proof}
Denote the right-hand side set in the above equation by $S_{\text{pos-viol}}$. For $\model \in \modelsivedge$, $P_{\model}\Paren{X,Y \mid \doop{Z=z}} = P_{\model}\Paren{X,Y \mid Z=z}$ if $P_{\model} \in S_{\text{pos-viol}}$. Further, in that case, $P_{\model}\Paren{X,Y \mid \doop{Z}}$ satisfies \eqref{eq:iv_notes} if and only if $P_{\model}\Paren{X,Y \mid Z}$ satisfies \eqref{eq:iv_notes}. This implies, 
$$\prt_2 \cap S_{\text{pos-viol}}^{\complement} = \emptyset.$$

This implies $\prt_2\subseteq S_{\text{pos-viol}}$. Let $P_{\model} \in S_{\text{pos-viol}}$ and define $z: \cX \times \cY \mapsto \cZ$ as $z(x,y) \triangleq \arg \max\limits_{z' \in \cZ} K\Paren{x,y \mid z'}$. Equation~\eqref{eq:iv_notes} can be expressed as 
\begin{equation*}
    \sum_y K\Paren{x^*,y \mid z(x^*,y)} \leq 1
\end{equation*}
where $x^* \triangleq \arg \max\limits_{x' \in \cX} \sum_y K\Paren{x',y \mid z(x',y)}$. 

If there exists $P_{\model_1} \in S_{\text{pos-viol}}$ such that $P_{\model_1}$ satisfies \eqref{eq:iv_notes}, then there exists some $y' \in \cY$ such that $z(x^*,y') = z$ where $x^* \triangleq \arg \max\limits_{x' \in \cX} \sum_y P_{\model_1} \Paren{x',y \mid\doop{z(x',y)}}$ and where w.l.o.g. we assume that $z \in \cZ$ is such that $P_{\model_1}(Z=z) = 0$. Define $\model_2 \in \modelsivedge$ such that $f_{\model_2}(X,Y,Z,U_X,U_Y,U_Z,U) = \bm{1}\Brack{X=x^*,Y=y'}\Paren{x^*,y',z} + (1-\bm{1}\Brack{X=x^*,Y=y'})f_{\model_1}(X,Y,Z,U_X,U_Y,U_Z,U)$ and the rest of the components in the tuple defining $\model_2$ are identical to that of $\model_1$. This implies that $P_{\model_2}\Paren{x^*,y' \mid \doop{z}} = 1$ implying that \eqref{eq:iv_notes} is not satisfied by $P_{\model_2}(X,Y \mid \doop{Z})$ but since $P_{\model_2}\Paren{Z=z} = 0$, induces the same observational distribution $P_{\model_2}$ as $P_{\model_1}$ implying $P_{\model_1} = P_{\model_2} \in \prt_2$. 

Assume to the contrary that no such $P_{\model_1} \in S_{\text{pos-viol}}$ exists such that $P_{\model_1}$ satisfies \eqref{eq:iv_notes}. Let $\model_1$ be any such model such that $P_{\model_1} \in S_{\text{pos-viol}}$ does not satisfy \eqref{eq:iv_notes} and assume that $P_{\model}(Z=z) = 0$. Define $\model_2$ such that $f_{\model_2}\Paren{
(X,Y,Z,U_X,U_Y,U_Z,U)} = f_{\model_1}\Paren{(X,Y,z',U_X,U_Y,U_Z,U)}$ where $z' \neq z$ (note $z'$ is unique since $|\cZ| = 2$) and the rest of the components of the tuple defining $\model_2$ are same as $\model_1$. Therefore, $P_{\model_2}\Paren{X,Y \mid \doop{z'}} = P_{\model_2}\Paren{X,Y \mid \doop{z}}$ for $z' \neq z$. This implies that the left hand side of \eqref{eq:iv_notes} evaluates to $1$ satisfying the IV inequalities. Since $P_{\model_2}(Z=z) = 0$, $P_{\model_1} = P_{\model_2} \in \prt_2$. Therefore, $\prt_2 = S_{\text{pos-viol}}$. 
\end{proof}
\fi
\ifdefined\SINGLE
\section{Proofs}\label{app:pf}
\else\fi
\subsection{Proofs for Propositions in Section~\ref{sec:prelim}}\label{app:sec:ermakov}

We prove Propositions~\ref{prop:consistencyfsigma},~\ref{prop:unifclopen} and ~\ref{prop:unifclosed} by stating the version in \citet{Ermakov17} first in both cases. Remember that we assume $\world \triangleq \nll \cup \alt$ to be compact in the topology of setwise convergence.

\begin{proposition}[Theorem 4.4 in \citet{Ermakov17}]\label{prop:consistencyfsigma_ermakov}If $\alt$ is contained in a $\sigma$-compact set in the topology of setwise convergence on $\cP\Paren{\cX}$, the following are equivalent: 
 \begin{enumerate}
     \item[(i)]there exists a binary test that is weakly consistent.
     \item [(ii)] $\nll$ and $\alt$ are contained respectively in disjoint $\sigma$-compact and $F_{\sigma}$ (i.e., countable union of closed sets) sets respectively.
 \end{enumerate}
 \end{proposition}

 \consistent*

 \begin{proof}
   Since $\world$ is compact, it is closed, implying that  $\bar{\nll}$ and $\bar{\alt}$ are closed subsets of $\world$ and therefore, are compact. This implies that both $\nll$ and $\alt$ are contained in $\sigma$-compact sets. If there exists a binary test that is weakly consistent, then according to Proposition~\ref{prop:consistencyfsigma_ermakov}, $\nll \subseteq A$, $\alt \subseteq B$ such that $A$ is $\sigma$-compact and $B$ is $F_{\sigma}$ and $A$ and $B$ are disjoint. Then $\alt = B \cap W$ is $F_{\sigma}$. Since $\nll$ is contained in a $\sigma$-compact set and there exists a weakly consistent binary test, according to Proposition~\ref{prop:consistencyfsigma_ermakov}, $\nll \subseteq B_0$, $\alt \subseteq A_0$ such that $A_0$ is $\sigma$-compact and $B_0$ is $F_{\sigma}$ and $A_0$ and $B_0$ are disjoint. Then $\nll = B_0 \cap W$ is $F_{\sigma}$. For the converse, if $\nll$ and $\alt$ are disjoint and $F_{\sigma}$ then since $\nll$ is a countable union of closed sets that are subsets of a compact set, $\nll$ is $\sigma$-compact implying that there exists a binary test that is weakly consistent. 
 \end{proof}

\begin{proposition}[Theorem 4.1 in \citet{Ermakov17}]\label{prop:unifclopen_ermakov}
 If $\nll, \alt$ are relatively compact in the topology of setwise convergence on $\cP\Paren{\cX}$, then the following are equivalent:
 \begin{enumerate}
     \item[(i)]there exists a binary test that is uniformly consistent.
     \item [(ii)] $\nll$ and $\alt$ are such that their closures in the topology of setwise convergence on $\cP\Paren{\cX}$ are disjoint.
 \end{enumerate}
 \end{proposition}

\unifclopen*

 \begin{proof}
     Since $\world$ is compact, it is closed, implying that  $\bar{\nll}$ and $\bar{\alt}$ are closed subsets of $\world$ and therefore, are compact. This implies that $\nll$ and $\alt$ are relatively compact. If there exists a binary test that is uniformly consistent, by Proposition~\ref{prop:unifclopen_ermakov}, since $\bar{\nll}$ and $\bar{\alt}$ are disjoint and $\nll$ and $\alt$ are also disjoint, both $\nll$ and $\alt$ are closed in the weak topology on $\cP\Paren{\cX}$. If $\nll$ and $\alt$ are closed in the weak topology on $\cP\Paren{\cX}$, since $\world$ is compact implying the weak topology and the topology of setwise convergence coincide, their closures are disjoint in the topology of setwise convergence on $\cP\Paren{\cX}$. 
 \end{proof}

 \begin{proposition}[Theorem 4.3 in \citet{Ermakov17}]\label{prop:unifclosed_ermakov}
 If $\nll$ is relatively compact and $\alt$ is contained in a set that is $\sigma$-compact in the topology of setwise convergence, then the following are equivalent:
 \begin{enumerate}
     \item[(i)]there exists a weakly consistent binary test that is uniformly consistent with respect to $\nll$.
     \item [(ii)] $\nll$ and $\alt$ are contained in respectively disjoint closed and $F_{\sigma}$ sets in the topology of setwise convergence on $\cP\Paren{\cX}$.  
 \end{enumerate}
 \end{proposition}

\unifclosed*
 \begin{proof}
     Since $\world$ is compact, it is closed implying that $\bar{\nll}$ and $\bar{\alt}$ are closed subsets of $\world$ and therefore, are compact. This implies that $\nll$ and $\alt$ are relatively compact. Further, $\alt \subset \world$ which is compact and also $\sigma$-compact. If (i) holds, by Proposition~\ref{prop:unifclosed_ermakov}, $\nll \subseteq A_0$ and $\alt \subseteq B_0$ such that $A_0$ is closed, $B_0$ is $F_{\sigma}$, and $A_0$ and $B_0$ are disjoint. Then $\nll = A_0 \cap \world$ is closed and $\alt = B_0 \cap \world$ is $F_{\sigma}$ in the weak topology on $\cP\Paren{\cX}$. Further, if $\nll$ and $\alt$ are closed and $F_{\sigma}$ in the weak topology on $\cP\Paren{\cX}$, then they are contained in disjoint closed and $F_{\sigma}$ sets in the topology of setwise convergence on $\cP\Paren{\cX}$ since $\world$ is compact.
 \end{proof}

\section{Decision Trees for Constructing Two-Stage Ternary Tests}\label{app:dt}
\begin{figure}[h]
\centering
\begin{subfigure}{0.45\linewidth}
\centering
    \begin{tikzpicture}[
    grow=down,
    level 1/.style={sibling distance=5cm, level distance=1.5cm},
    level 2/.style={sibling distance=1.5cm, level distance=1.5cm},
    edge from parent/.style={draw, -latex},
    every node/.style={rectangle, draw, rounded corners, inner sep=5pt, align=center},
    leaf/.style={fill=blue!10, draw=blue!80, font=\bfseries}
]
\node {$|I|=0$}
child { 
        node {$\prt_{onc}$ is open} 
        child{ node[leaf] {SN test with \\
        $\prt_{1A} = \prt_{onc}$\\$\prt_{2N}$ - either}}
        }
        child { 
        node {No open sets} 
        child { node[leaf] {None of the $12$ tests}}
        };
\end{tikzpicture}\caption{$|I|=0$}\label{fig:dt_0}
\end{subfigure}
\begin{subfigure}{0.40\linewidth}
\centering
    \begin{tikzpicture}[
    grow=down,
    level 1/.style={sibling distance=5cm, level distance=3cm},
    level 2/.style={sibling distance=1.5cm, level distance=4cm},
    edge from parent/.style={draw, -latex},
    every node/.style={rectangle, draw, rounded corners, inner sep=5pt, align=center},
    leaf/.style={fill=blue!10, draw=blue!80, font=\bfseries}
]
\node {$|I|=3$}
        child { node[leaf] {Any of the $12$ tests} };

\end{tikzpicture}\caption{$|I|=3$}\label{fig:dt_3}
\end{subfigure}

\vspace{10mm} 
\ifdefined\SINGLE
\begin{subfigure}{0.95\textwidth}\hspace{-30mm}
\begin{tikzpicture}[
    grow=down,
     level 1/.style={sibling distance=5.4cm, level distance=1.9cm},
    level 2/.style={sibling distance=3.1cm, level distance=2.2cm},
    edge from parent/.style={draw, -latex},
    every node/.style={rectangle, draw, rounded corners, inner sep=3.8pt, align=center},
    leaf/.style={fill=blue!10, draw=blue!80, font=\bfseries}
]
\else
\begin{subfigure}{0.95\textwidth}\hspace{-24mm}
\begin{tikzpicture}[
    grow=down,
     level 1/.style={sibling distance=5.3cm, level distance=1.9cm},
    level 2/.style={sibling distance=3cm, level distance=2.2cm},
    edge from parent/.style={draw, -latex},
    every node/.style={rectangle, draw, rounded corners, inner sep=3.8pt, align=center},
    leaf/.style={fill=blue!10, draw=blue!80, font=\bfseries}
]\fi
\node {$|I|=1$}
child{
node {$\prt_{cno},\prt_{onc1},\prt_{onc2}$}
child{node[leaf] {SN test with \\ $\prt_{1A} =\prt_{onc1}$ or $\prt_{onc2}$\\ $\prt_{2N} = \prt_{cno}$}
}
child{
node {SA test with \\ $\prt_{1N} = \prt_{cno}$\\ $\prt_{2N}$ - either}
}}
child{
node {$\prt_{cno},\prt_{onc},\prt_{nonc}$}
child
{node[leaf] {SN test with \\ $\prt_{1A} = \prt_{onc}$\\ $\prt_{2N}=\prt_{cno}$}
}
child
{
node{SA test with \\ $\prt_{1N} = \prt_{cno}$\\ $\prt_{2N}$ - either}
}
}
child
{node {$\prt_{co},\prt_{nonc1},\prt_{nonc2}$}
child
{node[leaf] {SN test with \\ $\prt_{1A} = \prt_{co}$\\ $\prt_{2N}$ - either}
}
child
{
node[leaf] {SA test with \\ $\prt_{1N} = \prt_{co}$\\ $\prt_{2N}$ - either}
}
}
child
{ node {$\prt_{cno},\prt_{nonc1}, \prt_{nonc2}$}
child{
node[leaf] {SA test with \\ $\prt_{1N}=\prt_{cno}$\\
$\prt_{2N}$ - either}
}
}
;
\end{tikzpicture}\caption{$|I|=1$}\label{fig:dt_1}
\end{subfigure}

\vspace{6mm}

\begin{subfigure}{0.95\textwidth}
\raggedleft
\ifdefined\SINGLE
\begin{tikzpicture}[
    grow=down,
     level 1/.style={sibling distance=8cm, level distance=1.5cm},
    level 2/.style={sibling distance=3.3cm, level distance=2cm},
    edge from parent/.style={draw, -latex},
    every node/.style={rectangle, draw, rounded corners, inner sep=5pt, align=center},
    leaf/.style={fill=blue!10, draw=blue!80, font=\bfseries}
]
\else
\begin{tikzpicture}[
    grow=down,
     level 1/.style={sibling distance=8cm, level distance=1.5cm},
    level 2/.style={sibling distance=3cm, level distance=2cm},
    edge from parent/.style={draw, -latex},
    every node/.style={rectangle, draw, rounded corners, inner sep=5pt, align=center},
    leaf/.style={fill=blue!10, draw=blue!80, font=\bfseries}
]
\fi
\centering
\node {$|I|=2$}
child{
node{$\prt_{cno1},\prt_{cno2},\prt_{onc}$ or\\
$\prt_{cno1},\prt_{cno2},\prt_{nonc}$}
child{node[leaf] {SN test with \\ $\prt_{1A} =\prt_{onc}$ or $\prt_{nonc}$\\ $\prt_{2N}$ - either}
}
}
child
{ node {$\prt_{co},\prt_{cno}, \prt_{nonc}$ or \\$\prt_{co},\prt_{c}, \prt_{o}$}
child{
node[leaf] {SA test with \\ $\prt_{1N}=\prt_{co}$\\
$\prt_{2N}=\prt_{cno}$}
}
child{
node[leaf] {SN test with \\ $\prt_{1A}=\prt_{co}$\\
$\prt_{2N}=\prt_{cno}$}
}
child{
node{SN test with \\ $\prt_{1A}=\prt_{nonc}$ or $\prt_{onc},$\\
$\prt_{2N}$- either}
}
}
;
\end{tikzpicture}\caption{$|I|=2$}\label{fig:dt_2}
\end{subfigure}\caption{Decision Trees where $I$ represents the closed sets. Subscript of $\prt$ indicates topological property: cno - closed and not open, onc - open and not closed, nonc - neither closed nor open, co - closed and open. All leaf nodes provide the desired error control expected from Theorem~\ref{thm:completeness_2stt}. The tests suggested by the blue colored leaf nodes provide error control for additional errors.}\label{fig:dt}
\end{figure}
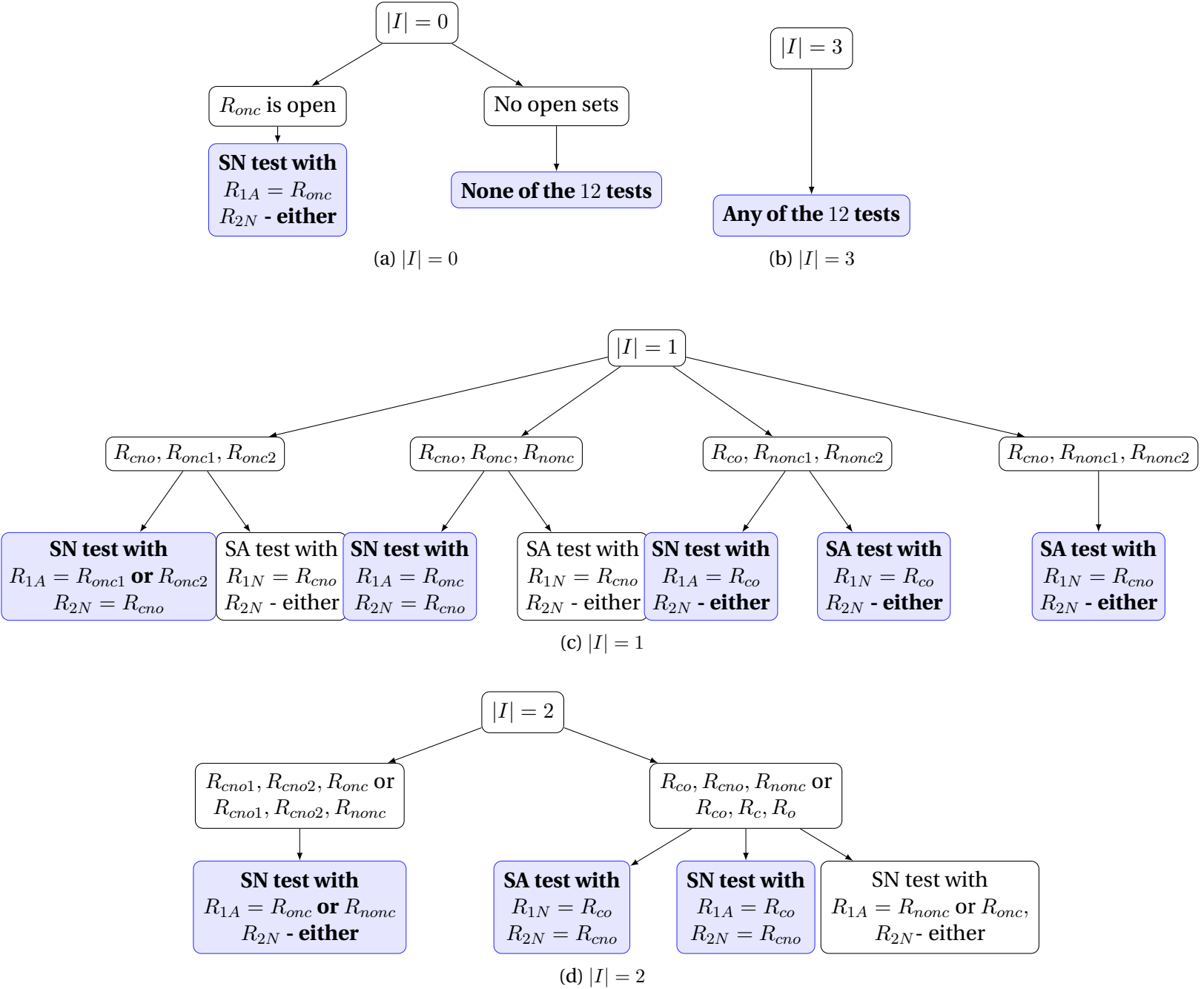